\documentclass[letterpaper]{amsart}
\usepackage{url}
\usepackage{color}
\usepackage{array}
\usepackage{xy}
\usepackage{cite}
\usepackage{setspace}
\usepackage{multicol}
\usepackage[pdftex]{graphicx}
\usepackage{amsmath,amsfonts,amsthm,amssymb,multirow}
\usepackage{algorithmic}
\usepackage{floatpag}
\usepackage{graphicx}
\usepackage[ruled,vlined, norelsize]{algorithm2e}

\newtheoremstyle{slanted}
{3pt}
{3pt}
{\slshape}
{}
{\bfseries}
{.}
{.5em}
{}
\theoremstyle{slanted}
\newtheorem{thm}{Theorem}[section]
\newtheorem{lem}[thm]{Lemma}
\newtheorem{prop}[thm]{Proposition}

\newtheorem{defn}[thm]{Definition}
\theoremstyle{remark}
\newtheorem{rem}[thm]{Remark}
\newtheorem{ex}[thm]{Example}
\usepackage{caption,subcaption}

\def \github {\url{github.com/williamkuszmaul/patternavoidance}}
\def \killpos {\text{\textbf{killpos}}}
\def \insertpos {\text{\textbf{insertpos}}}

\def \setpos {\text{\textbf{setpos}}}

\def \popcount {\text{\textbf{popcount}}}
\def \ctz {\text{\textbf{ctz}}}
\def \E {\mathbb{E}}
\def \pidd {$\Pi\text{DD}$}
\def \st {\text{st}}

\newcommand{\match}[3]{
\item {#1\newline Appears #2 times.\\ Example match: #3 \newline}
}

\newcommand{\fig}[3]{
  \begin{center}
  \begin{figure}
    #1
    \caption{#3}
    \label{#2}
  \end{figure}
  \end{center}
}

\newcommand{\subfig}[2]{
    \begin{subfigure}[h]{.48\linewidth}
      \begin{center}
        #1
      \end{center}
      \caption{#2}
    \end{subfigure} 
}

\def \tablepercentonavoiders{
\begin{tabular}{l | l  l  l  l }
$n \backslash k$ & 3 & 4 & 5 & 6 \\ \hline
8 & 0.777 & 0.932 & 0.953 & 0.945 \\
9 & 0.787 & 0.946 & 0.969 & 0.970 \\
10 & 0.797 & 0.956 & 0.978 & 0.983 \\
11 & 0.806 & 0.963 & 0.984 & 0.989 \\
12 & 0.814 & 0.969 & 0.988 & 0.993 \\
13 & 0.822 & 0.973 & 0.991 & 0.995 \\
14 & 0.829 & 0.977 & 0.993 & 0.996 \\
15 & 0.836 & 0.980 & 0.994 & 0.997 \\
16 & 0.842 & 0.982 & 0.995 & 0.998 \\
\end{tabular}}

\def \tablebruteavsingle{
\begin{tabular}{l | l  l  l  l }
$n \backslash k$ & 3 & 4 & 5 & 6 \\ \hline
8 & 0.001 & 0.013 & 0.016 & 0.012 \\
9 & 0.004 & 0.063 & 0.094 & 0.072 \\
10 & 0.009 & 0.299 & 0.832 & 0.958 \\
11 & 0.037 & 2.377 & 9.530 & 13.518 \\
12 & 0.151 & 19.068 & 112.187 & 198.764 \\
13 & 0.615 & 153.8 & 1348.32 & 3032.45 \\
\end{tabular}}
\def \tablebruteavmultiple{
\begin{tabular}{l | l  l  l  l }
$n \backslash k$ & 3 & 4 & 5 & 6 \\ \hline
10 & 0.009 & 0.026 & 0.049 & 0.062 \\
11 & 0.037 & 0.124 & 0.258 & 0.362 \\
12 & 0.151 & 0.572 & 1.345 & 2.120 \\
13 & 0.623 & 2.607 & 6.838 & 11.945 \\
14 & 2.490 & 11.801 & 34.316 & 66.623 \\
15 & 10.155 & 53.014 & 169.297 & 359.042 \\
16 & 41.299 & 236.709 & 822.06 & 1906.53 \\
\end{tabular}}
\def \tablebrutecntsingle{
\begin{tabular}{l | l  l  l  l }
$n \backslash k$ & 3 & 4 & 5 & 6 \\ \hline
8 & 0.021 & 0.020 & 0.013 & 0.007 \\
9 & 0.258 & 0.265 & 0.187 & 0.120 \\
10 & 3.361 & 3.763 & 2.791 & 1.940 \\
11 & 46.973 & 57.216 & 44.352 & 32.621 \\
12 & 705.082 & 930.591 & 752.467 & 581.081 \\
\end{tabular}}
\def \tablebrutecntmultiple{

  \begin{tabular}{l | l  l  l }
$n \backslash k$ & 3 & 4 & 5 \\ \hline
8 & 0.021 & 0.046 & 0.054 \\
9 & 0.258 & 0.650 & 0.911 \\
10 & 3.363 & 9.818 & 16.149 \\
11 & 46.960 & 156.46 & 297.638 \\
12 & 704.189 & 2646.83 & 5746.63 \\
\end{tabular}}

\def \tablehybridavsingle{
\begin{tabular}{l | l  l  l  l }
$n \backslash k$ & 3 & 4 & 5 & 6 \\ \hline
8 & 0.000 & 0.006 & 0.010 & 0.008 \\
9 & 0.001 & 0.030 & 0.059 & 0.051 \\
10 & 0.003 & 0.100 & 0.337 & 0.452 \\
11 & 0.009 & 0.654 & 3.388 & 5.559 \\
12 & 0.035 & 4.573 & 35.202 & 72.320 \\
13 & 0.131 & 32.533 & 378.392 & 985.548 \\
\end{tabular}}

\def \tablehybridavmultiple{
\begin{tabular}{l | l  l  l  l }
$n \backslash k$ & 3 & 4 & 5 & 6 \\ \hline
10 & 0.002 & 0.008 & 0.021 & 0.034 \\
11 & 0.009 & 0.035 & 0.102 & 0.181 \\
12 & 0.035 & 0.145 & 0.480 & 0.968 \\
13 & 0.131 & 0.598 & 2.237 & 5.043 \\
14 & 0.489 & 2.476 & 10.306 & 25.922 \\
15 & 1.825 & 10.193 & 47.311 & 130.265 \\
16 & 6.841 & 42.052 & 212.918 & 643.981 \\
\end{tabular}}
\def \tablehybridcntsingle{
\begin{tabular}{l | l  l  l  l }
$n \backslash k$ & 3 & 4 & 5 & 6 \\ \hline
8 & 0.003 & 0.004 & 0.003 & 0.002 \\
9 & 0.027 & 0.041 & 0.037 & 0.032 \\
10 & 0.286 & 0.468 & 0.440 & 0.400 \\
11 & 3.273 & 5.749 & 5.714 & 5.320 \\
12 & 40.809 & 77.063 & 80.786 & 77.323 \\
\end{tabular}}
\def \tablehybridcntmultiple{
\begin{tabular}{l | l  l  l }
$n \backslash k$ & 3 & 4 & 5 \\ \hline
8 & 0.003 & 0.010 & 0.018 \\
9 & 0.027 & 0.112 & 0.247 \\
10 & 0.286 & 1.400 & 3.616 \\
11 & 3.274 & 18.765 & 55.552 \\
12 & 40.716 & 269.727 & 893.669 \\
\end{tabular}}

\def \tableusavsingle{
\begin{tabular}{l | l  l  l  l }
$n \backslash k$ & 3 & 4 & 5 & 6 \\ \hline
8 & 0.000 & 0.000 & 0.001 & 0.001 \\
9 & 0.000 & 0.004 & 0.009 & 0.010 \\
10 & 0.001 & 0.019 & 0.045 & 0.050 \\
11 & 0.003 & 0.062 & 0.217 & 0.351 \\
12 & 0.008 & 0.339 & 1.779 & 3.590 \\
13 & 0.029 & 2.183 & 16.293 & 39.665 \\
\end{tabular}}

\def \tableusavmultiple{	
\begin{tabular}{l | l  l  l  l }
$n \backslash k$ & 3 & 4 & 5 & 6 \\ \hline
10 & 0.000 & 0.000 & 0.001 & 0.001 \\
11 & 0.002 & 0.002 & 0.003 & 0.003 \\
12 & 0.008 & 0.008 & 0.011 & 0.013 \\
13 & 0.029 & 0.031 & 0.039 & 0.046 \\
14 & 0.103 & 0.110 & 0.140 & 0.163 \\
15 & 0.368 & 0.396 & 0.504 & 0.589 \\
16 & 1.314 & 1.432 & 1.822 & 2.128 \\
\end{tabular}}

\def \tableuscntsingle{
\begin{tabular}{l | l  l  l  l }
$n \backslash k$ & 3 & 4 & 5 & 6 \\ \hline
8 & 0.003 & 0.002 & 0.003 & 0.002 \\
9 & 0.027 & 0.026 & 0.026 & 0.027 \\
10 & 0.285 & 0.302 & 0.302 & 0.309 \\
11 & 3.520 & 3.657 & 3.666 & 3.766 \\
12 & 42.741 & 44.752 & 44.791 & 45.716 \\
\end{tabular}}

\def \tableuscntmultiple{
\begin{tabular}{l | l  l  l }
$n \backslash k$ & 3 & 4 & 5 \\ \hline
8 & 0.003 & 0.004 & 0.006 \\
9 & 0.027 & 0.041 & 0.055 \\
10 & 0.286 & 0.453 & 0.637 \\
11 & 3.554 & 5.735 & 8.359 \\
12 & 42.842 & 74.717 & 110.991 \\
\end{tabular}}

\def \tableVavsingle{
\begin{tabular}{l | l  l  l  l }
$n \backslash k$ & 3 & 4 & 5 & 6 \\ \hline
8 & 0.000 & 0.004 & 0.007 & 0.007 \\
9 & 0.001 & 0.020 & 0.042 & 0.040 \\
10 & 0.003 & 0.070 & 0.212 & 0.289 \\
11 & 0.007 & 0.396 & 2.024 & 3.326 \\
12 & 0.028 & 2.665 & 20.160 & 41.060 \\
13 & 0.102 & 18.101 & 201.086 & 519.022 \\
\end{tabular}}

\def \tableVavmultiple{
\begin{tabular}{l | l  l  l  l }
$n \backslash k$ & 3 & 4 & 5 & 6 \\ \hline
8 & 0.000 & 0.000 & 0.005 & 0.032 \\
9 & 0.000 & 0.002 & 0.019 & 0.130 \\
10 & 0.002 & 0.009 & 0.081 & 0.552 \\
11 & 0.007 & 0.038 & 0.339 & 2.371 \\
12 & 0.027 & 0.148 & 1.416 & 10.288 \\
13 & 0.101 & 0.586 & 5.938 & 44.756 \\
\end{tabular}}

\def \tablepermlabavsingle{
\begin{tabular}{l | l  l  l  l }
$n \backslash k$ & 3 & 4 & 5 & 6 \\ \hline
8  & 0.018 & 0.019  & 0.029 & 0.034 \\
9  & 0.016 & 0.047 & 0.122 & 0.151 \\
10 & 0.024 &  0.147 & 0.581 & 0.757 \\
11 & 0.051 & 0.915 & 3.980 & 6.795 \\
12 & 0.123 & 5.020 & 35.127 & 74.387 \\
13 & 0.286 & 30.549 & 333.422 &  911.032 \\
\end{tabular}}

\def \tablePiDDavsingle{
\begin{tabular}{l | l  l  l  l }
$n \backslash k$ & 3 & 4 & 5 & 6 \\ \hline
8 & 0.011 & 0.007 & 0.011 & 0.009 \\
9 & 0.016 & 0.013 & 0.028 & 0.009 \\
10 & 0.027 & 0.034 & 0.067 & 0.039 \\
11 & 0.046 & 0.099 & 0.239 & 0.151 \\
12 & 0.087 & 0.361 & 0.952 & 0.965 \\
13 & 0.149 & 1.640 & 5.310 & 6.423 \\
14 & 0.219 & 6.434 & 24.810 & 34.897 \\
15 & 0.561 & 24.339 & 115.127 & 199.916 \\
16 & 1.672 & 91.030 & 567.907 & 1254.01 \\
\end{tabular}}

\def \tablePiDDcntsingle{
\begin{tabular}{l | l  l  l  l }
$n \backslash k$ & 3 & 4 & 5 & 6 \\ \hline
8 & 0.098 & 0.078 & 0.031 & 0.030 \\
9 & 0.263 & 0.249 & 0.087 & 0.040 \\
10 & 1.918 & 3.329 & 1.062 & 0.191 \\
11 & 15.638 & 40.400 & 17.453 & 3.671 \\
12 & 105.241 & 532.328 & 249.606 & 58.236 \\
\end{tabular}}

\def \tablePiDDcntmultiple{
\begin{tabular}{l | l  l }
$n \backslash k$ & 3 & 4 \\ \hline
8 & 0.095 & 0.137 \\
9 & 0.269 & 1.612 \\
10 & 1.824 & 20.139 \\
11 & 15.336 & 236.632 \\
\end{tabular}}

\newcommand{\ignore}[1]{}

\begin{document} 

\date{}
\title[]{Fast Algorithms for Finding Pattern Avoiders and Counting Pattern Occurrences in Permutations}  \author{William Kuszmaul}
\maketitle
\vspace{-.9 cm}
\begin{center}
Stanford University \\
\emph{kuszmaul@stanford.edu}
\end{center}




\begin{abstract}
Given a set $\Pi$ of permutation patterns of length at most $k$, we
present an algorithm for building $S_{\le n}(\Pi)$, the set of
permutations of length at most $n$ avoiding the patterns in $\Pi$, in
time $O(|S_{\le n - 1}(\Pi)| \cdot k + |S_{n}(\Pi)|)$. Additionally,
we present an $O(n!k)$-time algorithm for counting the number of
copies of patterns from $\Pi$ in each permutation in
$S_n$. Surprisingly, when $|\Pi| = 1$, this runtime can be improved to
$O(n!)$, spending only constant time per permutation. Whereas the
previous best algorithms, based on generate-and-check, take
exponential time per permutation analyzed, all of our algorithms take
time at most polynomial per outputted permutation.

If we want to solve only the enumerative variant of each problem, computing
$|S_{\le n}(\Pi)|$ or tallying permutations according to
$\Pi$-patterns, rather than to store information about every
permutation, then all of our algorithms can be implemented in
$O(n^{k+1}k)$ space.

Using our algorithms, we generated $|S_5(\Pi)|, \ldots, |S_{16}(\Pi)|$
for each $\Pi \subseteq S_4$ with $|\Pi| > 4$, and analyzed OEIS
matches. We obtained a number of potentially novel pattern-avoidance
conjectures.

Our algorithms extend to considering permutations in any set closed
under standardization of subsequences. Our algorithms also partially
adapt to considering vincular patterns.
\end{abstract}

\section{Introduction}


Over the past thirty years, the study of permutation patterns has
become one of the most active topics in enumerative
combinatorics. Given a pattern $\pi \in S_k$ and a permutation
$\tau \in S_n$, a \emph{$\pi$-hit} or \emph{copy of $\pi$} in $\tau$
is a $k$-letter subsequence of $\tau$ order-isomorphic to $\pi$. For
example, $857$ is a $312$-hit in $18365472$
(Figure \ref{fighitexample}). If $\tau$ contains no $\pi$-hits, we say
that $\tau$ \emph{avoids} $\pi$ and is in $S_n(\pi)$. Moreover, for a
set of patterns $\Pi$, $S_n(\Pi) =
\cap_{\pi \in \Pi}S_n(\pi)$.

Permutation patterns were first introduced in 1968, when Donald Knuth
characterized the stack-sortable $n$-permutations as exactly those
avoiding $312$, of which there are the Catalan number
$C_n$ \cite{Knuth68}. In 1985, Simion and Schmidt began a systematic
study of the combinatorial structures of $S_n(\Pi)$ for $\Pi \subseteq
S_3$ \cite{simion85}. Since then, permutation patterns have found
applications throughout combinatorics, as well as in computer science,
computational biology, and statistical mechanics \cite{kitaev11}. In
addition to the combinatorial structures of $\Pi$-hits being of
interest for individual $\Pi$, researchers have worked to build a more
general theory. The most famous result is the former Stanley-Wilf Conjecture,
posed in the 1980s independently by Richard Stanley and Herbert Wilf,
and proven in 2004 by Marcus and Tardos, which prohibits $|S_n(\Pi)|$
growing at a more than exponential rate \cite{marcus04}. Other work
has focused on characterizing when two sets $\Pi_1$ and $\Pi_2$ are
\emph{Wilf-equivalent}, meaning that $|S_n(\Pi_1)| = |S_n(\Pi_2)|$ for
all $n$ \cite{backelin07, kitaev11}.

Unfortunately, running large-scale experiments involving permutation
patterns is generally regarded as quite difficult \cite{albert01}. In
particular, detecting whether a pattern $\pi$ appears in a permutation
$w$ is NP-hard \cite{bose98}. In this paper, however, we will
circumvent this problem by detecting not whether $\pi$ appears in a
single permutation $w$, but instead finding the $\pi$-hits in large
collections of permutations, allowing us to obtain algorithms which
run in polynomial (and sometimes even constant) time per
permutation. In contrast, the best previously known algorithms, based
on generate-and-check, run in exponential time per permutation.

Significant research has already been conducted towards finding a fast
algorithm for determining whether $\tau \in S_n(\pi)$, which we will
refer to as the \emph{PPM} problem.

\begin{center}
\textbf{Permutation Pattern Matching Problem (PPM):} Given $\tau \in S_n$
and $\pi \in S_k$, determine whether $\tau \in S_n(\pi)$.
\end{center}

In 1998, Bose, Buss, and Lubiw showed that PPM is NP-hard in
general \cite{bose98}. Since then, research on PPM algorithms has
traveled down two paths, the first to find an exponential-time
algorithm with a small exponent, and the second to find fast PPM
algorithms for special cases of $\pi$. Notable progress in the first
direction includes an $O(1.79^n\cdot nk)$ algorithm due to Bruner and
Lackner
\cite{bruner12}, and a $2^{O(k^2\log k)}\cdot n$ algorithm due to
Guillemot and Marx \cite{guillemot14}. Notable progress in the second
direction includes polynomial-time algorithms when $\pi$ is separable
\cite{bose98, ibarra97, albert01, yugandhar05, han10}; an easily
parallelized linear-time algorithm when $|\pi| = 4$ \cite{han14, albert01}; and
an algorithm whose runtime depends on a natural complexity-measure of
$\pi$, running fast for $\pi$ with small complexity-measure
\cite{ahal08}. Additionally, results have been found for more general
types of patterns such as vincular patterns \cite{bruner13}.

For experimental research purposes, however, most permutation-pattern computations involve
not just one permutation, but many. Indeed, the two most common
computations are to build all of $S_{\le n}(\pi)$, or to count
copies of $\pi$ in each $\tau \in S_n$.

\begin{center}
\textbf{Permutation Pattern Avoiders Problem (PPA):} Given a
permutation $\pi \in S_k$ and $n \in \mathbb{N}$, construct all
permutations of size at most $n$ that avoid the pattern $\pi$.
\end{center}

\begin{center}
\textbf{Permutation Pattern Counting Problem (PPC):} Given a
permutation $\pi \in S_k$ and $n \in \mathbb{N}$, find the number of
copies of $\pi$ in each permutation of size at most $n$.
\end{center}

One common approach to PPA and PPC, which we will refer to as
\emph{generate-and-check}, is to iterate through candidate
permutations and apply PPM to each candidate \cite{albert01, PermLab,
  sage}. However, recent algorithms introduced by Inoue, Takahisa, and
Minato take a different approach, representing sets of permutations in
highly compressed data structures called \pidd's, and then using
\pidd-set-operations to solve PPA and PPC \cite{inoue14,
  inoue14followup}. Although the asymptotic nature of their algorithms
is unknown due to the enigmatic compression performance of \pidd's,
their algorithms experimentally run much faster than the
generate-and-check approach.

In this paper, we introduce the first provably fast algorithms for PPA
and PPC. Surprisingly, PPC can be solved in $\Theta(n!)$ time (Theorem
\ref{thmfactorial}), spending only amortized constant time per
permutation despite $\pi$ appearing $n! {n \choose k}/{k!}$ times as a
pattern in $S_n$. Similarly, PPA can be solved in $O(|S_{\le
n-1}(\pi)|\cdot k + |S_n(\pi)|)$ time, spending linear time per output
permutation. Our algorithms are the first proven to spend
sub-exponential time per output permutation.

In Section \ref{secmem}, for the enumerative versions of PPA and PPC,
we show how to implement both algorithms in $O(n^{k + 1}k)$ space,
making them practical even for very large computations on small
machines.

Both algorithms extend to considering a set of patterns $\Pi$ (of
possibly varying lengths), rather than just a single pattern
$\pi$. Interestingly, their runtimes depend only on $k
= \max_{\pi \in \Pi}|\pi|$, building $S_n(\Pi)$ in time $O(|S_{\le n -
1}(\Pi)| \cdot k + |S_n(\Pi)|)$ (Theorem \ref{thmusavoid2}) and
counting $\Pi$-patterns in each $\tau$ in $S_n$ in time $O(n!\cdot k)$
(Theorem
\ref{thmdownset}). Additionally, our algorithms easily adapt to
finding avoiders and counting copies of patterns in $\Pi$ in arbitrary
downsets of permutations -- for example, efficiently finding
the separable permutations which are $\Pi$-avoiders. We also partially
extend our results to when $\pi$ is a vincular pattern.

Our algorithms open new doors for data-driven research studying the
structure of permutation classes. Previously daunting large-scale
computations are now easily within reach. For example, our software
can generate $|S_1(\Pi)|, \ldots, |S_{16}(\Pi)|$ for every $\Pi
\subseteq S_4$ (regardless of $|\Pi|$) in just under twenty-five
minutes on our Amazon C3.8xlarge machine\footnote{Run in parallel with
  hyperthreading enabled for a total of 36 hardware threads. Our code
  is parallelized using Cilk.}. A brief analysis of the resulting
number sequences reveals that hundreds of OEIS sequences seemingly
previously unaffiliated with pattern avoidance can be used to
enumerate $|S_n(\Pi)|$ for some $\Pi$. Moreover, when we filter out
the matches which can be proven using insertion encoding techniques,
we are left with 32 OEIS sequences, matching with 289 sets of
patterns, each of which appears to represent a novel and unsolved
conjecture.




The layout of this paper is as follows. In Section \ref{secprelim}, we
introduce (mostly standard) conventions. In Section \ref{secalg1}, we
introduce and analyze a simplified version of our PPA algorithm, which
is then refined in Section \ref{secalg2}, and extended to PPC in
Section \ref{secalg3}. In Section \ref{secmem}, we modify our
algorithms to achieve good space utilization. In
Section \ref{sectest}, we compare our algorithms (running in serial)
experimentally to the best alternatives.  In
Section \ref{secconjectures}, we use our algorithms to run large-scale
computations, automatically generating hundreds of conjectures, some
of which seem quite interesting.  Finally, Section
\ref{secconclusion} concludes with directions of future work and some results on vincular patterns.

\section{Definitional Preliminaries}\label{secprelim}
In this section, we set conventions for the paper. We begin by
discussing pattern avoidance.

\begin{defn}
  A \emph{permutation in $S_n$} is a word containing each letter from $1$ to $n$ exactly once.
\end{defn}

\begin{defn}
  Given a word $\tau$ of $n$ distinct letters, the \emph{standardization}
  $\st(\tau)$ is the permutation $\sigma \in S_n$ such that $\tau_i < \tau_j$ exactly
  when $\sigma_i < \sigma_j$.
\end{defn}

\begin{ex}
  The standardization of $5397$ is $\st(5397) = 2143$.
\end{ex}

\begin{defn}
  Two words $\tau_1$ and $\tau_2$ are order-isomorphic if $\st(\tau_1) = \st(\tau_2)$.
\end{defn}

\begin{defn}
  Let $\pi \in S_k$ and $\tau \in S_n$. A \emph{$\pi$-hit} is any subword
  of $\tau$ order-isomorphic to $\pi$. On the other hand, $\tau$
  \emph{avoids the pattern} $\pi$ if $\tau$ has no $\pi$-hits
\end{defn}

\begin{ex}
  An example $123$-hit in $18365472$ is the subword $367$, while an
  example $312$-hit is the subword $857$. These hits are shown
  graphically in Figure \ref{fighitexample}. Observe however, that
  there is no $3124$-hit in $18365472$. Thus $18365472$ avoids the
  pattern $3124$.
\end{ex}

\begin{figure}
  \begin{center}
\begin{tabular}{c c c}
  \includegraphics[scale=.4]{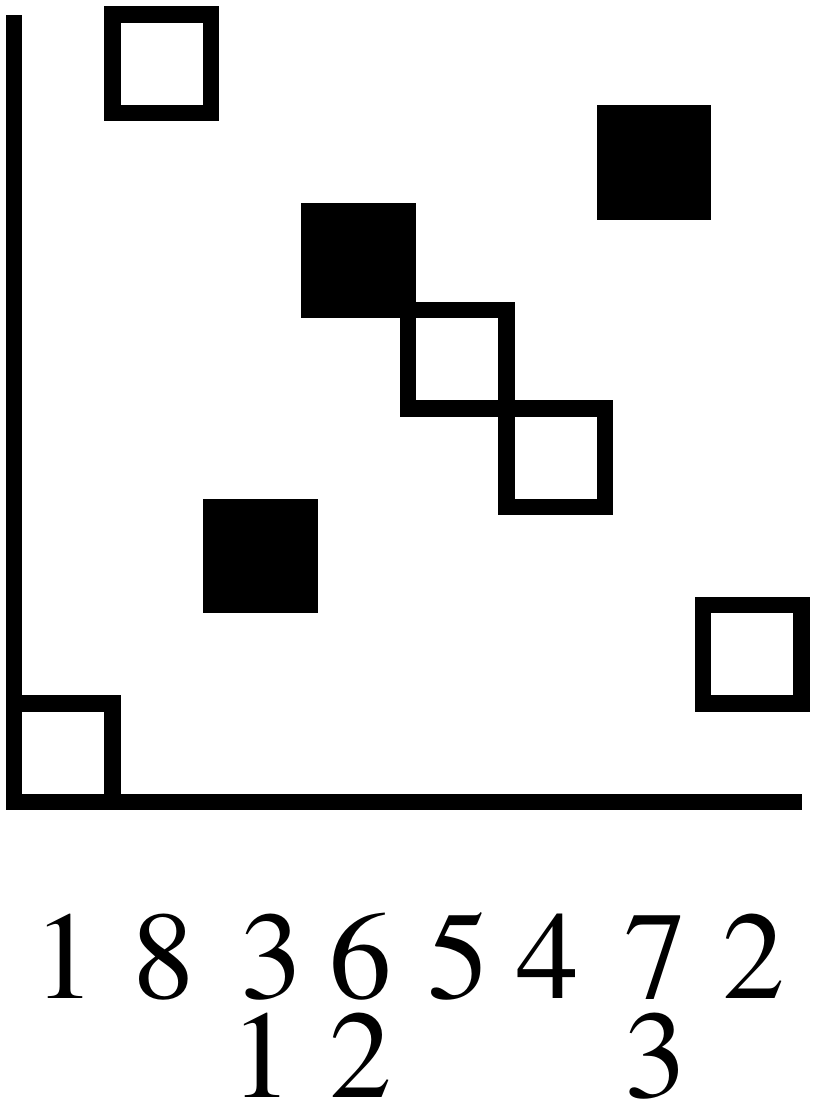} & \hspace{1 cm} &
  \includegraphics[scale=.4]{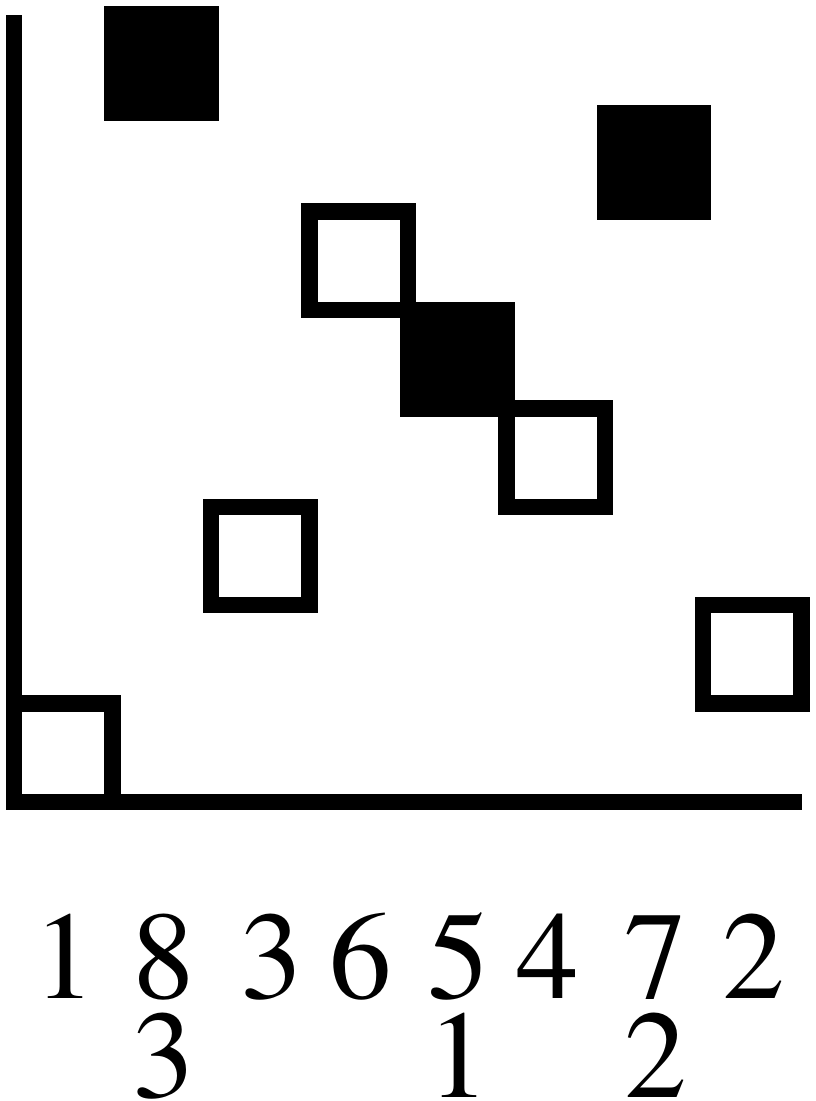}
\end{tabular}
\end{center}
\caption{Example $123$-hit and $312$-hit in 18365472. In this figure,
  a permutation is represented graphically.  A square is placed at
  position $i,j$ when the $i$th element of the permutation is $j$.  In
  the left figure the subword 367 is shown to form a $123$-hit, and in
  the right figure the subword 857 is shown to form a $312$-hit.}
\label{fighitexample}
\end{figure}

Similarly, if $\Pi$ is a set of permutations, then the $\Pi$-hits are
just the $\pi$-hits for each $\pi \in \Pi$. And a permutation $\tau \in
S_n$ avoids $\Pi$ if it has no $\Pi$-hits. In this context, $\Pi$ may
be referred to as a \emph{set of patterns}, and we say that $\tau$
\emph{avoids the patterns} in $\Pi$.

Next, we introduce common short-hands for sets which we will study.
\begin{defn}
We use $S_{\le n}$ to denote $S_1 \cup S_2 \cup \cdots \cup S_n$.  
\end{defn}

\begin{defn}
  Let $\pi$ (resp. $\Pi$) be a pattern (resp. set of patterns), and
  $D$ be a set. Then $D(\pi)$ (resp. $D(\Pi)$) is the subset of $D$
  which avoids $\pi$ (resp. $\Pi$).
\end{defn}

\begin{ex}
  Since $S_n$ is the set of permutations of size $n$, the set $S_n(123)$ is
  the set of permutations of size $n$ with no increasing subsequence
  of length three.
\end{ex}

Our algorithms will build data about permutations up from data about
smaller permutations. Consequently, they are designed to work on
\emph{downsets} of permutations.

\begin{defn}
  A set of permutations $D$ is a \emph{downset} if for all $\tau \in D$, for all non-empty subwords $\tau'$ of $\tau$, $\st(\tau') \in D$.
\end{defn}

Examples of downsets include $S_{\le n}$, the permutations with $j$ or
fewer inversions (for a constant $j$), the permutations with $j$ or
smaller major index, the permutations avoiding a given set of
patterns, and the separable permutations. Additionally, the unions and
the intersections of downsets are also downsets.

Next, we introduce notation for obtaining from a permutation $\tau$ a new
permutation that is either one smaller or one larger in size.

\begin{defn}
  Given $\tau \in S_n$ and $i \in \{1, \ldots, n+1\}$, we define $\tau
  \uparrow^i$ to be the permutation obtained by inserting $n+1$ to be
  in the $i$-th position of $\tau$.
\end{defn}

\begin{defn}
  Given $\tau \in S_n$ and $i \in \{1, \ldots, n\}$, we define $\tau
  \downarrow _i$ to be the standardization of the word obtained by
  removing the letter $(n-i+1)$ from $\tau$.
\end{defn}

\begin{ex}
For example, $13524\uparrow^2 = 163524$, while $13524\downarrow_2 = \st(1352) = 1342$.
\end{ex}

Note that $\uparrow^i$ and $\downarrow_i$ are not inverses. Whereas
$\uparrow^i$ inserts a letter into the $i$-th position, $\downarrow_i$
removes the $i$-th largest-valued letter. Though subtle, these
distinctions will play a critical role in the optimizations presented
in Section \ref{secalg2}.

It will often be useful to refer to the word formed by the largest
$k$-letters of a permutation as the $k$\emph{-upfix} of the
permutation. For example, the $3$-upfix of $15234$ is $534$.






\section{PPA in time polynomial per avoider}\label{secalg1}

In this section, we introduce the key ideas for obtaining an
asymptotically fast algorithm to build $S_{\le n}(\Pi)$. Combined,
these ideas yield a simple algorithm running in time $O(S_{\le n -
1}(\Pi) n^2k)$, the first algorithm to spend only polynomial time per
$\Pi$-avoiding permutation. This algorithm can additionally be adapted
to build $D(\Pi)$ for a downset $D$ (assuming constant-time membership
queries for $D$.) In later sections, we will introduce techniques for
reducing the polynomial term and for achieving good space bounds.

Our algorithm relies fundamentally on a simple observation which
transforms pattern detection into a dynamic programming
problem. Whereas detecting whether a permutation $\tau \in S_n$ contains
a pattern $\pi \in S_k$ naively takes time $O\left({n \choose k}k\right)$,
Proposition \ref{propfundamental} shows how to perform the same
computation in polynomial time using information about smaller
permutations.

\begin{prop}
Let $\Pi$ be a set of patterns, each of length at most $k$, and let
$\tau$ be a permutation length $n$. Pick $X$ to be any set of at least
$\min(k + 1, n)$ distinct entries of $\tau$. Then $\tau$ lies in
$S_n(\Pi)$ if and only if the following two conditions hold.
\begin{enumerate}
\item $\tau \notin \Pi$ and
\item for each entry $x \in X$, the standardization of $\tau$ with the entry $x$ removed lies in $S_{n-1}(\Pi)$.
\end{enumerate}
\end{prop}\label{propfundamental}
\begin{proof}
  Suppose $\tau \in S_n(\Pi)$. Then Condition (1) holds trivially, and
  Condition (2) holds because $S_{\le n}(\Pi)$ is a downset.

  On the other hand, suppose Conditions (1) and (2) hold. Observe that
  if the standardization of $\tau$ with the letter $x$ removed lies in
  $S_{n-1}(\Pi)$, then any $\Pi$-hit in $\tau$ must use the letter
  $x$. Thus Condition (2) implies that any $\Pi$-hit in $\tau$ must
  use at least $\min(k+1, n)$ distinct letters of $\tau$. If $k < n$,
  this is impossible, since the longest pattern in $\Pi$ is length at
  most $k$. If $k \ge n$, then $\tau$ can only contain a $\Pi$-hit if
  that $\Pi$-hit comprises all of $\tau$, a contradiction by Condition
  (1).
\end{proof}

\begin{ex}
In Figure \ref{fundamentalex}, we apply
Proposition \ref{propfundamental} to $25143$ and to $34215$ in order
to determine whether each avoids $123$. For each permutation, we
remove its first, second, third, and fourth letters, standardize the
result, and record whether it avoids $123$. Assuming that we have
already computed which $4$-letter permutations avoid $123$, this
entire process takes polynomial time for each permutation.

Because all four tests pass for $25143$, we conclude that it avoids
the pattern $123$. On the other hand, $34215$ fails two tests and
contain a $123$ pattern.

The decision to remove each the first four letters was arbitrary,
since Proposition \ref{propfundamental} allows us to use any four
letters. In fact, our actual algorithms will always use the letters
$n, n-1, \ldots, (n - \max(k + 1, n) + 1)$ when testing for
avoidance. Although unmotivated for the time being, this decision will
make optimizations in Section \ref{secalg2} easier to discuss.
\end{ex}

\begin{figure}[h]
  \begin{center}
  \begin{tabular}{c  c}
  \begin{tabular}{| l | r | r |}
 \hline  Letter removed & Permutation in $S_4$                     & Avoids 123? \\ \hline
      First letter:  & \ {\color{red}2}5143               &  \\
                   & \ \phantom{2}4132                   &  yes \\ \hline 
      Second letter:  & \ 2{\color{red}5}143               &     \\               
                   & \ 2\phantom{5}143                   &  yes \\ \hline 
      Third letter:  & \ 25{\color{red}1}43               &     \\
                   & \ 14\phantom{1}32                   &  yes \\ \hline 
      Fourth letter:  & \ 251{\color{red}4}3               &     \\                 
                   & \ 241\phantom{4}3                   &  yes \\ \hline 
  \end{tabular}
&
  \begin{tabular}{| l | r | r |}
 \hline   Letter removed & Permutation in $S_4$            & Avoids 123? \\ \hline
       First letter:  & \ {\color{red}3}4215               &  \\
                   & \ \phantom{3}3214                   & yes \\ \hline 
       Second letter:  & \ 3{\color{red}4}215               &     \\               
                   & \ 3\phantom{4}214                   & yes \\ \hline 
       Third letter:  & \ 34{\color{red}2}15               &    \\
                   & \ 23\phantom{2}14                   & no \\ \hline 
       Fourth letter:  & \ 342{\color{red}1}5               &     \\                 
                   & \ 231\phantom{1}4                   & no \\ \hline 
  \end{tabular}

\end{tabular}
  \caption{Applying Proposition \ref{propfundamental} to determine whether $25143$ and whether $34215$ avoid the pattern $123$.}
  \label{fundamentalex}
    \end{center}
 
\end{figure}

Armed with Proposition \ref{propfundamental} we can now derive a fast
algorithm. The simplest algorithm for building $S_{\le n}(\Pi)$ is to
brute-force check whether each permutation $\tau$ in $S_{\le n}$ is
$\Pi$-avoiding. If we do this by checking every $|\pi|$-subsequence of
$\tau$ for each $\pi \in
\Pi$, this takes time
$$ O\left(\sum_{\pi \in \Pi}n! {n \choose |\pi|} |\pi| \right).$$

This formula becomes simpler if $\Pi$ comprises $l$ permutations of
size $k$. In this case, the algorithm runs in time $O\left(n!  \cdot
{n \choose k} kl\right)$.

Our first task is to shrink the $n!$ term. Observe that $S_{\le n} (\Pi)$ is
a downset. Consequently, every element in $S_n(\pi)$ can be
obtained by inserting $n$ into some position of a permutation in $S_{n
- 1}(\Pi)$. Thus we can build $S_n(\Pi)$ from $S_{n-1}(\Pi)$ by
checking pattern-avoidance for each permutation $\tau$ obtained by
inserting $n$ into some position of an element in $S_{n -
1}(\Pi)$. Since there are at most $|S_{n - 1}(\Pi)| \cdot n$ such $\tau$, this yields
an algorithm which generates $S_{\le n}(\Pi)$ in time $$O\left(|S_{\le
{n-1}}(\Pi)| \cdot n \cdot {n \choose k} \cdot kl\right).$$

Our next task is to shrink the ${n \choose k}$ and eliminate the
dependence on $l$. Recall that proposition \ref{propfundamental} shows
that if $S_{n-1}(\Pi)$ is already computed, then checking whether
$\tau \in S_n(\Pi)$ for some $\tau \in S_n$ can be achieved in $O(kn)$
time, rather than in $O\left({n \choose k} \cdot kl\right)$ time. In
particular, to see that $\tau \in S_n(\Pi)$ we need only check that
$\tau \not\in \Pi$ and that $\tau\downarrow_{i} \in S_{n-1}(\Pi)$ for
each $i \in [\min(k + 1, n)]$ (Algorithm \ref{algavoidance}). This
brings our total runtime down to $O(|(S_{\le n-1}(\Pi)| \cdot n^2k)$.
Note that the number of patterns in $\Pi$ does not increase the time
needed to detect whether a permutation is $\Pi$-avoiding.

\begin{thm}\label{thmbasicalg}
Let $\Pi$ be a set of patterns and $k = \max_{\pi \in \Pi} |\pi|$.
   The set $S_{\le n}(\Pi)$ can be constructed in $O(|(S_{\le
   n-1}(\Pi)| \cdot n^2k)$ time.
\end{thm}
\begin{proof}
  By Proposition \ref{propfundamental}, this is accomplished through
  Algorithm \ref{algbuildavoiders}. Note that one can easily obtain
  each $\tau \downarrow_i$ from $\tau$ in $O(n)$ time.
\end{proof}

\begin{rem}
  Note that for single patterns $\pi$, we have $|S_n(\pi)| \le
  |S_{n+1}(\pi)|$ for all $n$. In particular, depending on $\pi$, one
  of the maps $\tau \rightarrow \tau\uparrow^1$ or
  $\tau \rightarrow \tau \uparrow ^{n+1}$ is an injection from
  $S_n(\pi)$ to $S_{n+1}(\pi)$. Thus for a single pattern, our
  algorithm is efficient even if we only want to compute $S_n(\pi)$,
  with runtime $O(|S_n(\pi)| \cdot n^3k)$, which using results from
  the next section can be reduced to $O(|S_n(\pi)| \cdot nk)$.

  However, $|S_n(\Pi)| \le |S_{n+1}(\Pi)|$ need not be true when $|\Pi|
  > 1$. For example, if $\Pi$ contains the increasing pattern of
  length $a$ and the decreasing pattern of length $b$, then by the
  Erd{\"o}s-Szekeres Theorem, no permutation of length greater than
  $(a+1)(b+1)+1$ is $\Pi$-avoiding  \cite{erdos35}.
\end{rem}

\begin{algorithm}
  \KwIn{Hash table $H$ such that $H \cap S_{n-1} = S_{n-1}(\Pi)$, Hash table $\Pi$, $k:= \max_{\pi \in \Pi}{|\pi|}$, Permutation $\tau \in S_n$}
  \KwOut{Whether $\tau \in S_n(\Pi)$}
  \If{$\tau \in \Pi$} {
    \Return{ false}\;
  }
  \For{$ i \in \{1, \ldots, \min(k+1, n)\} $} {
    \If{$\tau\downarrow_i \not\in H$}{
      \Return{false}\;
    }
  }
  \Return{true}\;
  \caption{\textbf{DetectAvoider}}
  \label{algavoidance}
\end{algorithm}

\begin{algorithm}
  \KwIn{Hash table $\Pi$, $k:= \max_{\pi \in \Pi}{|\pi|}$, $n$}
  \KwOut{A hash table containing $S_{\le n}(\Pi)$}
  UnorderedSet Avoiders\;
  Queue Unprocessed\;
  \If{$1 \not\in \Pi$}
     {Unprocessed.enqueue(1)\;
      Avoiders.add(1)\;}
  \While{not Unprocessed.empty()} {
    Perm := Unprocessed.dequeue()\;
    \For{$i \in \{1, \ldots, Perm.size() + 1\}$} {
      NewPerm := Perm$\uparrow^i$\;
      \If{DetectAvoiders(Avoiders, $\Pi$, $k$, NewPerm)} {
        Avoiders.insert(NewPerm)\;
        \If{NewPerm.size() $< n$} {
          Unprocessed.enqueue(NewPerm)\;
         }
      }
    }
  }
  \Return{Avoiders\;}
  \caption{\textbf{BuildAvoiders}}
  \label{algbuildavoiders}
\end{algorithm}

Observe that Algorithm \ref{algbuildavoiders} can be easily modified
to generate $S_{\le n}(\Pi) \cap D$ for downsets $D$, assuming
membership in $D$ can be determined in constant time. In particular,
prior to checking whether NewPerm is an avoider, we throw out NewPerm
if it is not in $D$. In fact, using the optimized version of
Algorithm \ref{algbuildavoiders} which will be presented in
Section \ref{secalg2} (Theorem \ref{thmusavoid2}), we can build
$D(\Pi)$ in time $O(|D(\Pi) \cap S_{\le n - 1}|n)$. An example
candidate for $D$ is the set of permutations in $S_{\le n}$ with $j$
or fewer inversions for a fixed $j$; in particular, by keeping track
of the inversion statistic for permutations in UnprocessedQueue, one
can detect when NewPerm has inversion statistic greater than $j$ in
constant time.\footnote{In this case, a clever implementation could
further reduce the time to $O(|D(\Pi)| \cdot k)$ by only
considering $\text{Perm}\uparrow^i$ for values of $i$ large enough to
keep the number of inversions below $j$.}

Other examples of downsets include the separable permutations,
and the permutations with major index at most a fixed
constant. Recently, the study of permutation avoidance with respect to
permutation statistics such as major index and inversion number have
become of particular interest \cite{claesson12, opler15}.

\section{Optimizations for PPA}\label{secalg2}

In the preceding section, we presented
Algorithm \ref{algbuildavoiders} which builds $S_{\le n}(\Pi)$ in time
$O(|(S_{\le n-1}(\Pi)| \cdot n^2k)$. In this section, we introduce two
optimizations, each of which reduces the runtime by a factor of $n$,
bringing the total runtime down by a factor of $n^2$ to $O(|S_{\le
n-1}(\Pi)| \cdot k + |S_n(\Pi)|)$. The first optimization relies on
encoding permutations as integers, allowing permutation operations to
be performed using bit manipulations. The second optimization performs
pattern detection on multiple permutations at once, leading to
additional speedup.

Because $S_n$ and $S_n(\pi)$ grow quickly, foreseeable applications of
our algorithms are likely to use permutations that can be easily
stored in a few machine words. Consequently, we assume that words can
be stored as integers, with the $i$-th $j$-bit block representing the
$i$-th letter for some fixed $j$ (which we call the \emph{block-size};
words may not contain a letter larger than $2^j$). Using this
assumption, we can shave off a factor of $n$ from
Algorithm \ref{algbuildavoiders}'s runtime.

\begin{thm}
By representing permutations as integers,
   Algorithm \ref{algbuildavoiders} can be implemented to run in time
   $O(|S_{\le n-1}(\Pi)| \cdot nk)$.
   \label{thmintermediateavoid}
\end{thm}
\begin{proof}
The analysis from Theorem \ref{thmbasicalg} of
Algorithm \ref{algbuildavoiders} assumes that each computation of
$\tau\uparrow^i$ or $\tau\downarrow_i$ takes time $O(n)$. In this
analysis, we will show that in the context of
Algorithm \ref{algbuildavoiders}, and with a bit of extra bookkeeping,
these computations can each be reduced to constant time. In
particular, each $\tau\uparrow^i$ can be accomplished in constant time
using bit hacks, and each $\tau\downarrow_{i+1}$ can be obtained from
$\tau\downarrow_{i}$ using bit hacks and information about $\tau^{-1}$.

Note that the following operations are constant time for integers
representing a word $\tau$ stored as a permutation with block-size $j$:
$\tau(i)$, which returns the $i$-th letter of $\tau$; $\setpos(\tau, i, j)$,
which sets the $i$-th letter of $\tau$ to value $j$; $\insertpos(\tau, u,
v)$, which slides the final $n - u + 1$ letters of $\tau$ one position to
the right, and inserts the value $v$ in the $u$-th position; and
$\killpos(\tau, i)$, which slides the final $n - i$ letters of $\tau$ one to
position the left, erasing the $i$-th position. These are each easily
implemented using standard integer operations, including bit shifting,
which allows for multiplication and division by powers of two in
constant time. For example, if $\tau$ is an integer representing a word,
$$\killpos(\tau, i) = \tau \bmod 2^{j(i - 1)} + \lfloor \tau / 2^{ij} \rfloor 2^{j(i-1)},$$
which can be implemented in C as

$$\tau \& ((1 << (j*i - j)) - 1) + (\tau >> (i * j)) << (j * i - j).$$

Using these basic operations, if $\tau$ represents a permutation in
$S_n$, we can compute $\tau\uparrow^i = \insertpos(\tau, i, n+1)$ in
constant time.  We can compute $\tau \downarrow_{i+1}$ from
$\tau \downarrow_i$ and $\tau^{-1}$ (i.e., the integer representation of the
inverse permutation) in constant time by inserting $n-i$ into position
$\tau^{-1}(n - i + 1)$ of $\tau
\downarrow_i$, and then killing the $\tau^{-1}(n-i)$-th letter of the
result. Finally, we can also compute $(\tau\uparrow^{i+1})^{-1}$ from
$(\tau\uparrow^{i})^{-1}$ and $\tau$ in constant time, by incrementing the
$\tau(i)$-th position of $\tau \uparrow^i$ and decrementing the $(n+1)$-th
position. Consequently, for Algorithm \ref{algbuildavoiders}, all
computations of $\tau \uparrow^i$ and $\tau\downarrow_i$ can be performed in
constant time, as long as one also computes and stores
$(\tau\uparrow^i)^{-1}$ when computing $\tau \uparrow ^i$.

This reduces the runtime for Algorithm \ref{algbuildavoiders} from
$O(|S_{\le n-1}(\Pi)| \cdot n^2k)$, as derived in
Theorem \ref{thmbasicalg}, to $O(|S_{\le n-1}(\Pi)| \cdot nk)$, as
desired.
\end{proof}

Surprisingly, we can further optimize the algorithm to shave off
another linear factor. To do this, we must introduce the notion of
an \emph{extension map}.

\begin{defn}
Let $\tau \in S_n(\Pi)$. Let $I$ be the set of $i \in [n + 1]$ such that
$\tau \uparrow^i \in S_{n+1}(\Pi)$. Then the \emph{extension map
$\Psi^\Pi(\tau)$} of $\tau$ is the $(n+1)$-letter bit map with $i$-th
letter equal to 1 exactly when $i \in I$, and equal to 0 otherwise.
\end{defn}

\begin{ex}
Consider $12 \in S_2(123)$. Observe that $\Psi^{123}(12)$ is $110$
because inserting $3$ in either of the first two positions of $12$
results in another $123$-avoider but inserting $3$ in the third
position does not.
\end{ex}

\begin{defn}
Let $j \in [n]$ and $\tau \in S_n(\Pi)$. Let $I$ be the set of $i \in [n
+ 1]$ such that $\tau\uparrow^i \downarrow_{j + 1} \in S_n(\Pi)$.  Then
the \emph{$(n - j + 1)$-ignoring extension map $\Psi^\Pi_{n - j +
1}(\tau)$} of $\tau$ is the $(n+1)$-letter bit map with $i$-th letter
equal to one exactly when $i \in I$.
\end{defn}

\begin{ex}
Consider $53412 \in S_n(123)$. Then the $4$-ignoring extension map of
$53412$ tells us for which $i$ we can insert $6$ in position $i$ to
get a permutation whose only $123$-patterns involve the letter
$4$. Consequently, $\Psi^{123}_4(53412) = 111110$.
\end{ex}

The next theorem shows how to count $\Pi$-avoiders in only $O(k)$ time
per avoider. In addition to the integer operations traditionally used
in the RAM model, the algorithm uses two operations which most modern
machines implement in a single instruction. The first is $\popcount$,
which returns the number of $1$s in an integer's binary
representation. The second is $\ctz$, which returns the number of
trailing $0$-bits of an integer, starting at the least-significant bit
position.

\begin{thm}\label{thmusavoid2}
Let $\Pi$ be a set of patterns, the longest of which is length
$k$. The values $|S_1(\Pi)|, \ldots, |S_{n}(\Pi)|$ can be computed in
time $O(|(S_{\le n-1}(\Pi)| \cdot k)$. Moreover, in time $O(|(S_{\le
n-1}(\Pi)| \cdot k + |S_{n}(\Pi)|)$, one can construct $S_{\le
n}(\Pi)$.
\end{thm}
\begin{proof}
Our computational model allows us to store $O(n)$ bits in an
integer. As a result, we can store extension maps as unsigned
integers, allowing us to perform integer operations on them in
constant time.

Consider a $\Pi$-avoiding permutation $\tau \in S_m(\Pi)$ for some
$m \ge k$. (We will handle smaller $\tau$ later.) By
Proposition \ref{propfundamental}, $$\Psi^\Pi(\tau) = \wedge_{j \in [n
- k, n]}\Psi^\Pi_j(\tau),$$ where $\wedge$ denotes the \emph{and}
operator. (Call this Observation (1).)

Moreover, given $\tau^{-1}$ and $\Psi^\Pi(\tau\downarrow_{m - j +
1})$, we can compute $\Psi^\Pi_j(\tau)$ in constant time. (Call this
Observation (2).) In particular, since $\Psi^\Pi(\tau\downarrow_{m - j
+ 1})$ is the extension map of the standardization of $\tau$ with $j$
removed, and since $\Psi^\Pi_j(\tau)$ is the $j$-ignoring extension
map of $\tau$, we get the following relationship. For $i \in
[1, \tau^{-1}(j)]$, the $i$-th bit of $\Psi^\Pi_j(\tau)$ is the same
as that of $\Psi^\Pi(\tau\downarrow_{m - j + 1})$; and for $i \in
[\tau^{-1}(j) + 1, n + 1]$ the $i$-th bit of $\Psi^\Pi_j(\tau)$ equals
the $(i-1)$-th bit of $\Psi^\Pi(\tau\downarrow_{m - j + 1})$. Thus
$\Psi^\Pi_j(\tau)$ can be obtained from $\Psi^\Pi(\tau\downarrow_{m -
j + 1})$ by shifting bits in positions $\tau^{-1}(j)+1, \ldots, n+1$
to the right by one, and inserting a copy of the $\tau^{-1}(j)$-th bit
in the $(\tau^{-1}(j) + 1)$-th position.

Combining Observations (1) and (2), we can build $\{\Psi^\Pi(\tau)
: \tau \in S_{m}(\Pi)\}$ in time $O(|S_m(\Pi)| \cdot k)$ out of
$\{(\tau, \tau^{-1}) : \tau \in S_{m}(\Pi)\}$ and
$\{\Psi^\Pi(\tau): \tau \in S_{m - 1}(\Pi)$\}. If $m = n - 1$, then at
this point we can use the $\popcount$ instruction to to count the
number of on-bits appearing in extension maps of permutations in
$S_m(\Pi)$. This takes $O(|S_{n - 1}(\Pi)|)$ time and gives us a value
for $|S_n(\Pi)|$. If $m < n - 1$, then we want to build
$\{(\tau, \tau^{-1}) : \tau \in S_{\le m + 1}(\Pi)\}$ and then repeat
the entire process for $m + 1$.

From the extension maps of avoiders in $S_{m}$, we can obtain
$S_{m+1}(\Pi)$ in time $O(|S_{m + 1}(\Pi)|)$ by repeatedly taking
advantage of the $\ctz$ operation in order to extract the $1$-bit
positions from each map. Constructing $\{\tau^{-1}: \tau \in S_{m +
1}(\Pi)\}$ is not as easy however, and would take $O(|S_{m +
1}(\Pi)| \cdot n)$ time to do naively. We are saved, however, by the fact
that we only need each $\tau^{-1}$ to be correct in its largest $k$
values. Thus if we choose to only update these values, then we can
obtain the inverses in time $O(|S_{m + 1}(\Pi)| \cdot k)$.

At this point we have an algorithm which only starts to work once we
have already built the avoiders in $S_k$. In particular, Observation
(1), which states that
$$\Psi^\Pi(\tau) = \wedge_{j \in [n - k, n]}\Psi^\Pi_j(\tau),$$ may not hold
if $|\tau| < k$ (if $\tau\uparrow^i \in \Pi$, then the formula may falsely
identify $\tau\uparrow^i$ as an avoider). This is easily fixed, however,
by simply checking $\Pi$-membership for each detected avoider.
\end{proof}


\section{Counting Pattern Occurrences in $S_{\le n}$}\label{secalg3}

Building on the ideas in Sections \ref{secalg1} and \ref{secalg2}, in this section
we present a dynamic algorithm for counting $\Pi$-hits in each
permutation of $S_n$ in $O(n!k)$ time. Interestingly, when $|\Pi| =
1$, this can be improved to an $O(n!)$ time algorithm. Additionally,
given a preconstructed downset $D \subseteq S_{\le n}$ and the
inverses of each $\tau \in D$, our algorithm extends to run in $O(|S|k)$
time. The inverses for each $\tau \in D$ are required so that
$\tau\downarrow_1, \ldots, \tau\downarrow_{k+1}$ may be computed in $O(k)$
time (using the same technique as in
Theorem \ref{thmintermediateavoid}); recall, however, that they can be
obtained at no additional asymptotic cost if we build $D$ through
repeated applications of the $\uparrow^i$ operation.

For this section, fix $\Pi$ to be a set of patterns, $k = \max_{\pi
  \in \Pi} |\pi|$, and $n \in \mathbb{N}$. For permutations $\tau$, let
$P(\tau)$ denote the number of $\Pi$-hits in $\tau$.

\begin{defn}
  Let $P_i(\tau)$ be the number of $\Pi$-hits in $\tau$ containing the
  entire $i$-upfix of $\tau$. (Recall that the $i$-upfix of $\tau$
  refers to the $i$ largest-valued letters in $\tau$.)
\end{defn}

\begin{ex}
  Suppose $\tau = 1234$ and $\Pi = \{123\}$. Then $P_0(\tau) = 4$, $P_1(\tau) = 3$, $P_2(\tau) = 2$, $P_3(\tau) = 1$, and $P_4(\tau) = 0$.
\end{ex}

Observe that $P_0(\tau) = P(\tau)$. Surprisingly, whereas $P(\tau)$
satisfies no straightforward recurrence relation, $P_i(\tau)$
does. The following proposition can be thought of as a natural
extension of Proposition \ref{propfundamental} from the context of
pattern detection to the context of pattern counting.

\begin{prop}\label{lemPi}
  Let $\tau \in S_n$. Then

  \[ P_i(\tau) = \left\{
  \begin{array}{ll}
   P_{i+1}(\tau) + P_i (\tau \downarrow_{i + 1}) & \mbox{if $i<n$ and $i \le k$,} \\
   1                             & \mbox{if $i=n$ and $\tau\in\Pi$, and} \\
   0                             & \mbox{otherwise.}
  \end{array}
 \right\} \]
\end{prop}
\begin{proof}
  Suppose $i < n$ and $i \le k$. Then the $\Pi$-hits in $\tau$ using
  $\tau$'s entire $(i+1)$-upfix are counted by $P_{i+1}(\tau)$. And the
  $\Pi$-hits in $\tau$ using $\tau$'s entire $i$-upfix but not $\tau$'s entire
  $(i+1)$-upfix are counted by $P_i(\tau\downarrow_{i+1})$.

  Suppose $i = n$. Then the $i$-upfix of $\tau$ forms a pattern in $\Pi$
  if and only if $\tau \in \Pi$.

  Finally, if $i > k$ or $i > n$ then $P_i(\tau) = 0$. In particular, if $i > k$,
  then no pattern in $\Pi$ can use all of the first $i$ letters of
  $\tau$, since $k = \max_{\pi \in \Pi}|\pi|$.
\end{proof}

Given a permutation $\tau$ and its inverse $\tau^{-1}$, and using the
optimizations introduced in Theorem \ref{thmintermediateavoid},
Proposition
\ref{lemPi} yields an $O(k)$ algorithm to compute each $P_i(\tau)$ for a
permutation in terms of each $P_i(\tau')$ for smaller permutations
$\tau'$ (Algorithm \ref{algcounthits})\footnote{Recall $\tau^{-1}$ is
needed for fast computation of $\tau\downarrow_i$ for
$i \in \{1, \ldots, k+1\}$.}. Note that Algorithm \ref{algcounthits}
treats each $P_i$ as a globally accessible hash table mapping
permutations to integers, and that Algorithm \ref{algcounthits}
assumes access to $\Pi$ and $k$.

\begin{algorithm}
  \KwIn{Permutation $\tau \in S_n$} 
  \KwOut{Assigns values to $P_i(\tau)$ for each $i \in \{0, \ldots, k+1\}$}
  $P_{k+1}(\tau):= 0$\;
  \For{$i \in \{k, \ldots, 0\}$} {
    $P_i(\tau) := 0$\;
    \If{$i = n$ and $\tau \in \Pi$}{$P_i(\tau):= 1$}
    \If{$i < n$}{
      $P_i(\tau):=P_i(\tau\downarrow_{i+1}) + P_{i+1}(\tau)$\;
     }
   }
  \caption{\textbf{Count($\tau$):} Counting $\Pi$-hits in $\tau$.}
  \label{algcounthits}
\end{algorithm}

\begin{thm}\label{thmdownset}
  Given a downset $D \subseteq S_{\le n}$, and the inverse of
  each $d \in D$, one can construct $P(\tau)$ for each $\tau \in D$ in
  $O(|D|\cdot k)$ time.
\end{thm}
\begin{proof}
  Given $D$, bucket-sort can be used to construct each of $D \cap
  S_{i}$ for $1 \le i \le n$ in $O(|D|)$ total time. One can then use
  Algorithm \ref{algcounthits} to compute $P(\tau)$ for each $\tau \in
  (D \cap S_i)$ for $i$ from $1$ to $n$ (as well as $P_i(\tau)$ for
  $O(k)$ different $i$). This takes $O(|D| \cdot k)$ time. Note that
  we are assuming each $\tau\downarrow_i$ in the algorithm takes
  constant time to compute; this is easily accomplished using the
  exact same technique as in Theorem \ref{thmintermediateavoid}, and
  is the reason we require the inverse of each $d \in D$.
\end{proof}

The algorithm in Theorem \ref{thmdownset} can also be adapted for
downsets $D \subseteq S_{\le n}$ for which set membership is
conditional on the number of $\Pi$-hits of a permutation.

For one important example of this, suppose $D$ is the set of
permutations in $S_{\le n}$ with $j$ or fewer $\Pi$-hits for some
fixed $j$. Since $D$ is a downset, every element in $S_n \cap D$ is of
the form $\tau\uparrow ^i$ for some $i \in \{1, \ldots n\}$ and
$\tau \in S_{n-1}
\cap D$. Thus we can build $S_n \cap D$ out of $S_{n-1} \cap D$ while
simultaneously using Proposition \ref{lemPi} to compute $P(\tau)$ for
each $\tau \in D$. To accomplish this, we use Algorithm
\ref{algcountlimitedhits} to identify whether a permutation $\tau$ is in
$S_n \cap D$ based on values of $P_i(\tau')$ for $\tau' \in S_{n-1} \cap
D$. At the same time, if Algorithm \ref{algcountlimitedhits} concludes
that a permutation is in $S_n \cap D$, it computes $P_i(\tau)$ for each
$i$.  In turn, Algorithm \ref{algbuildlimitedhits} uses Algorithm
\ref{algcountlimitedhits} to compute each $P_i(\tau)$ for all $\tau \in
D$. Observe that Algorithm \ref{algbuildlimitedhits} runs in $O(|D
\cap S_{\le n-1}|\cdot nk)$ time. In particular, for each
permutation $\tau$ in $D \cap S_{\le n-1}$, we run Algorithm
\ref{algcountlimitedhits} on each $\tau\uparrow^i$.

\begin{algorithm}
  \KwIn{HashTable $H$ such that $H \cap S_{n-1} = D \cap S_{n-1}$, Permutation $\tau \in S_n$, $j$} 
  \KwOut{Returns whether $\tau$ has $\le j$ $\Pi$-hits. If true, assigns values to $P_i(\tau)$ for each $i \in \{0, \ldots, k+1\}$}
  $P_{k+1}(\tau):= 0$\;
  \For{$i \in \{k, \ldots, 0\}$} {
    $P_i(\tau) := 0$\;
    \If{$i = n$ and $\tau \in \Pi$}{
      $P_i(\tau):= 1$\;
    }
    \If{$i < n$}{
      \If{$\tau\downarrow_{i+1} \not\in H$} {
        \For{$r \in \{k+1, \ldots, i+1\}$}{
          $P_{r}.remove(\tau)$\;
        }
        \Return{false\;}
       }
      $P_i(\tau):=P_i(\tau\downarrow_{i+1}) + P_{i+1}(\tau)$
     }
  }
  \If{$P_0(\tau) > j$} {
    \For{$i \in \{k+1, \ldots, 0\}$}{
      $P_i.remove(\tau)$\;
    }
    \Return{false\;}
  }
  \Return{true\;}
\caption{\textbf{CountHitsBounded} Counts $\Pi$-hits in $\tau$ if $\tau$
  has at most $j$ $\Pi$-hits (and is thus said to be in $D$); returns
  false if $\tau$ has more than $j$ $\Pi$-hits.}
\label{algcountlimitedhits}
\end{algorithm}

\begin{algorithm}
  \KwIn{$n$, $j$, $\Pi$}
  \KwOut{Returns set of permutations $\tau$ in $S_{\le n}$ with $\le j$ $\Pi$-hits; Also computes values of $P_i(\tau)$.}
  UnorderedSet D\;
  Queue Unprocessed\;
  \If{$1 \not\in \Pi$ or $j \ge 1$}
     {Unprocessed.enqueue(1)\;
      D.add(1)\;}
  \While{not Unprocessed.isempty()} {
    Perm := Unprocessed.dequeue()\;
    \For{$i \in \{1, \ldots, Perm.size + 1\}$} {
      NewPerm := Perm$\uparrow^i$\;
      \If{CountHitsBounded(D, NewPerm, j)} {
        D.insert(NewPerm)\;
        \If{NewPerm.size() $< n$} {
          Unprocessed.enqueue(NewPerm)\;
         }
      }
    }
  }
  \Return{D\;}
  \caption{\textbf{BuildPermsWithBoundedHits}}
  \label{algbuildlimitedhits}
\end{algorithm}

When $j = 0$, Algorithm \ref{algbuildlimitedhits} simply builds
$S_{\le n}(\Pi)$. In fact, in this case the algorithm can be cleaned
up to become Algorithm \ref{algbuildavoiders}.

Theorem \ref{thmdownset} allows us to count $\Pi$-hits in each $\tau \in
S_n$ in $O(n!k)$ time. Surprisingly, this can be improved even further
when $|\Pi| = 1$.
\begin{thm}\label{thmfactorial}
  Let $\pi \in S_k$. Then the number of $\pi$-hits in each $\tau \in S_n$
  can be computed in $\Theta(n!)$ time, regardless of $k$.
\end{thm}
\begin{proof}
  For a permutation $\tau$, let $i$ be the smallest $i$ such that the
  $i$-upfix of $\tau$ is not order-isomorphic to the $i$-upfix of
  $\pi$. Then $P_i(\tau) = 0$. Thus we can modify Algorithm
  \ref{algcounthits} to not bother computing $P_j(\tau)$ for $j > i$. In
  particular, $P_j(\tau)$ for $j > i$ will never be requested later in
  the algorithm; any $\tau'$ such that $\tau'\downarrow_{k} = \tau$ for some $k
  > i$ will also have its $i$-upfix not order-isomorphic to $\pi$'s.

  Note that given that the $(i-1)$-upfix of $\tau$ is order-isomorphic
  to the $(i-1)$-upfix of $\pi$, and using information about
  $\tau^{-1}$, one can check whether the $i$-upfix is as well in
  constant time. With this in mind, we can analyze our new algorithm.

  Let $T_r$ be the indicator function taking value 1 when the
  $r$-upfix of a permutation is order-isomorphic to $\pi$'s
  $r$-upfix. Then the new algorithm spends time proportional to $O(1)
  + \sum_r T_r(\tau)$ on each permutation $\tau$. However, $\E(T_r(\tau)) \le
  1/r!$ over all $\tau \in S_{\le n}$. Thus the algorithm runs in $O(n!)$
  time.
\end{proof}

In fact, we conjecture that the same trick reduces Algorithm
\ref{algavoidance} to an $O(|S_{n-1}(\pi)|n)$ time algorithm for any
pattern $\pi$. This would not necessarily reduce the runtime of the
enumeration algorithm in Theorem \ref{thmusavoid2} to
$O(|S_{n-1}(\pi)|)$, however, since the algorithm would still be
asymptotically bottle-necked by the updating of inverses. Regardless,
the hack can be added to both algorithms to reduce cache misses.

To prove the conjecture, one would show that $\E(T_r(\tau))$ is small for
$\tau$ from $S_n(\pi)\uparrow := \{\tau \uparrow ^i \mid i \in \{1,\ldots, n\},
\tau \in S_{n-1}(\pi)\}$. For example, when $\pi = 123\cdots k$, this can
be done as follows.  Suppose $\tau \in S_n(\pi)\uparrow$ has an
increasing $r$-upfix $p_1p_2\cdots p_r$. Since $\tau$'s $r$-upfix is
in increasing order and $\tau$ avoids $12\cdots k$, it follows that
any reordering of the letters in $\tau$'s $r$-upfix will also result
in a permutation avoiding $12\cdots k$. If we reorder the letters in
the $r$-upfix to be $p_2p_3p_4 \cdots p_i$ with $p_1$ inserted in some
position other than the first, then we see that we can match each
$\tau \in S_n(\pi)\uparrow$ having an increasing $r$-upfix with $r-1$
permutations, each in $S_n(\pi)\uparrow$ and each with an increasing
$(r-1)$-upfix (but not an increasing $r$-upfix). It follows that
$\E(T_r(\tau)) \le \E(T_{r-1}(\tau))/r$ over $\tau \in
S_n(\pi)\uparrow$, implying the conjecture for $\pi = 123\cdots k$. In
fact, proving $\E(T_r(\tau)) \le \E(T_{r-1}(\tau))/c$ for any constant
$c > 1$ would be sufficient, which is why the conjecture seems very
likely to be true in general.

\begin{rem}\label{rempopcount}
  In practice, the technique introduced in Theorem \ref{thmfactorial}
  is worth implementing even for large sets of patterns $\Pi$ (for
  both PPA and PPC). In order for this to be efficient, however, one
  needs to quickly identify whether the $i$-upfix of a permutation
  $\tau$ is an $i$-upfix of \emph{any} permutation $\pi \in \Pi$. An
  efficient technique for this, taking constant time per $i$-upfix, is
  discussed in Appendix A.
\end{rem}

\section{Eliminating the memory bottleneck and making parallelism easy}\label{secmem}

So far, our algorithms have required space nearly proportional to
their runtime. In this section we restructure our algorithms so that,
without changing their runtimes, we asymptotically reduce space usage
to at most $O(n^{k + 1}k)$. Consequently, our algorithms are practical
for even very large computations on small computers. At the same time,
these changes make our algorithms easily implemented in parallel.

Of course, if one wants to actually store $S_n(\Pi)$ or $P(\tau)$ for
each $\tau \in S_n$, then space efficiency is futile. However, in this
section, we assume that the goal is \emph{enumeration}, to either evaluate
 $|S_n(\Pi)|$ or to tally how many $\tau \in S_n$ have each value of
$P(\tau)$.

For this entire section, define $T$ to be the inclusion tree of all
permutations, meaning that a node $v$ has children $v \uparrow^i$ for
each $i \in [1, |v| + 1]$. For a node $v$, define $v$'s \emph{$j$-th
level children} $C^j(v)$ to be the set of nodes in the $(j + 1)$-th
level of the subtree of which $v$ is the root. In particular, these
are the permutations in $S_{|v| + j}$ whose smallest $|v|$ letters are
order-isomorphic to $v$. 

The following lemma will play a key role in improving memory
utilization. In particular, the recursions on which both our PPA and
PPC algorithms are based compute information about a given $\tau \in
S_n$ based only on information about $\tau\downarrow_i$ for each
$i \in \min (n, k + 1)$. Lemma \ref{lemmem} tells us that we can
therefore compute information about the $(k+1)$-th level children of
$v$ based only on information about the $k$-th level children of $v$.

\begin{lem}\label{lemmem}
  For a given a set of permutations $A$, define $A\downarrow_i$ as $\{s
\downarrow_i : s \in A\}$. Then for any node $v \in T$ and for any positive integers $i$ and $j$ satisfying $i \le j$, 
  $$C^{j}(v)\downarrow_i \subseteq C^{j- 1}(v).$$
\end{lem}
\begin{proof}
  The elements of $C^{j-1}(v)$ are precisely the permutations in
  $S_{|v| + j - 1}$ with $v$ as their $|v|$-downfix (i.e., the word
  formed by the letters $1, \ldots, |v|$ is order-isomorphic to
  $v$). Every element of $C^{j}(v)$ also has $|v|$-downfix
  $v$. Moreover, $|v| + j - i + 1 > |v|$ since $i \le j$, implying
  that every element of $C^j(v)\downarrow_i$ has $|v|$-upfix $v$ as
  well, completing the proof.
\end{proof}

Our approach to PPC in Section \ref{secalg3} uses $O((n-1)!k)$
space. In particular, we perform a breadth-first traversal of $T \cap
S_{\le n}$, using Proposition \ref{lemPi} at each node $v$ to compute
each $P_i(v)$. However, Lemma \ref{lemmem} suggests an
$O(n^{k+1}k)$-space approach.

\begin{thm}
Let $\Pi$ be a set of patterns, the longest of which is length $k$. We
can count permutations in $S_n$ by $\Pi$-hits in $O(n!k)$ time using
$O(n^{k+1}k)$ space.
\end{thm}
\begin{proof}
If $n < k$, then Algorithm \ref{algcounthits} is already
sufficient. Otherwise, we restructure the algorithm as follows.

First compute $P_i(\tau)$ for each $\tau \in S_{\le k}$ using the
Algorithm \ref{algcounthits} (in $O(k!)$ space). Then traverse $T
\cap S_{\le n - k}$ depth-first. When visiting a node $v$, compute
each $P_i(c)$ for each $c \in C^k(v)$. Observe that by
Lemma \ref{lemmem} this computation depends only on elements of
$C^{k}(v\downarrow_1)$, which we will have already computed due to the
depth-first nature of our computation\footnote{Note that as a base
case we consider $C^k$ of the empty permutation to be the permutations
of size $k$.}. Having computed each $P_i$ of each $c \in C^k(v)$, we
update our tally of how many permutations have each number of
$\Pi$-hits, and we then store each $P_i(c)$ (to be accessed while
visiting $v$'s children in $T$). However, when we return to $v$'s
parents during our depth-first traversal of $T$, we no longer need to
store these $P_i(c)$ values and we throw them out.

At a given point in the traversal of $T \cap S_{n - k}$, we may store
as many $(P_0, P_1, \ldots, P_k)$-tuples as $O(\sum_{i = 0}^{n - k}
(i+1)(i+2)\cdots(i+k))$, bounding our memory-usage at $O(n^{k + 1}k)$;
note that the space needed to keep tally of permutations by $\Pi$-hits
is bounded above by ${n \choose k} + {n \choose {k - 1}} + \cdots + {n
\choose 1} \le n^{k + 1}$.
\end{proof}

\begin{rem}
When $|\Pi| = 1$, we can use the technique from Theorem
\ref{thmfactorial} to obtain $O(n^{k+1})$-space usage, since instead
of storing entire $(k+1)$-tuples of $P_i$'s, we store on average a
constant number per permutation. Additionally, the technique brings
the runtime down to $O(n!)$.
\end{rem}

We can apply a similar optimization to the PPA algorithm introduced in
Theorem \ref{thmusavoid2}.

\begin{thm}
Let $\Pi$ be a set of patterns, the longest of which is length $k$. The values
$|S_1(\Pi)|, \ldots, |S_{n}(\Pi)|$ can be computed in time $O(|(S_{\le
n-1}(\Pi)| \cdot k)$ and space $O(n^k)$.
\end{thm}
\begin{proof}
If $n \le k - 1$, then Theorem \ref{thmusavoid2} is already
sufficient, requiring space no more than $n! \le n^k$. Otherwise, we
restructure the theorem's algorithm as follows.

As a base case, use the algorithm from Theorem \ref{thmusavoid2} to
count $\Pi$-avoiders in $S_{\le k - 1}(\Pi)$ and build the extension map
for each avoider in $S_{k- 1}(\Pi)$.

Computing the extension map for an avoider $v$ uses only the extension
maps of $v\downarrow_1, \ldots, v\downarrow_k$. Therefore, by
Lemma \ref{lemmem}, to build the extension maps for each permutation
in $C^{k-1}(v) \cap S_{\le n}(\Pi)$, it suffices to have stored the
extension maps for each permutation in $C^{k-1}(v \downarrow_1) \cap
S_{\le n}(\Pi)$. Thus after we have built the extension maps for $S_{k
- 1}(\Pi)$ as a base case, we can restructure the algorithm from
Theorem \ref{thmusavoid2} as follows. We perform a depth-first
traversal on $T \cap S_{n - k}$. When visiting a node $v$, build the
extension map of each permutation in $C^{k-1}(v) \cap S_{\le
n}(\Pi)$. Store these extension maps to be accessed later in the
depth-first traversal; upon returning to $v$'s parents, however, we
throw these extension maps away.

In order for this to not compromise the algorithm's runtime, there is
a slight subtlety. We need to be able to build $C^{k -1 }(v) \cap
S_{\le n}(\Pi)$ out of $C^{k - 1}(v\downarrow_1) \cap S_{\le n}(\Pi)$
in time $O(|C^{k - 1}(v \downarrow_1) \cap S_{\le n}(\Pi)| \cdot k +
|C^{k - 1}(v) \cap S_{\le n}(\Pi)|)$. To accomplish this, when
visiting $v\downarrow_1$ in the depth-first traversal, one partitions
$C^{k-1}(v \downarrow_1) \cap S_{\le n}(\Pi)$ according to the
position of $|v|$ relative to $1, 2, \ldots, |v| - 1$\footnote{Since
the algorithm in Theorem \ref{thmusavoid2} remembers the positions of
the final $k$ letters of the inverses of the permutations in
$C^{k-1}(v \downarrow_1) \cap S_{\le n}(\Pi)$, we can build this
partition in time $O(k|C^{k-1}(v \downarrow_1) \cap S_{\le
n}(\Pi)|)$.}. Then, when visiting $v$, one can build $C^{k-1}(v) \cap
S_{\le n}(\Pi)$ from the extension maps of elements of
$C^{k-1}(v \downarrow_1) \cap S_{\le n}(\Pi)$ having $|v|$ in the same
position relative to $1, 2, \ldots, |v|-1$ as in $v$\footnote{Recall
that the $\ctz$ operator can be used to quickly determine for which
positions an extension map takes value 1}. This resolves the issue,
allowing us to retain our original runtime.

At any given moment in the algorithm, at each depth of the
depth-first traversal, we store no more than $n^{k - 1}$ extension
maps. Thus the algorithm uses $O(n^k)$ space.
\end{proof}

In practice, if $C^k(v) \cap S_{\le n}(\Pi)$ is never very large for
any $v$, then the space usage for PPA may be much smaller than
$O(n^k)$. In particular, the expected memory consumption at a given
instance in the algorithm is $O(\sum_{j = k}^{n - 1}|S_j(\Pi)| / |S_{j
- k + 1}(\Pi)|)$, which by the former Stanley-Wilf Conjecture (proven
in \cite{marcus04}), grows at most linearly with $n$ (with a
potentially large constant depending on $\Pi$). Thus, a large machine
running many pattern-avoidance computations in parallel can treat
space usage as growing linearly with $n$.

In addition to reducing space-usage asymptotically, the optimizations
in this section make parallelizing our algorithm easy. Indeed, by
visiting multiple of a node's children at a time, the depth-first
traversals of $T$ can be parallelized without risking write-conflicts
for hash maps containing either extension maps or
$P_i$-values. Although we test our algorithms in serial in
Section \ref{sectest}, we have released a parallelized implementation
at \github, 
which we used for our computations in
Section \ref{secconjectures}.

\section{Implementations and Performance Comparisons}\label{sectest}
 
In this section, we test our algorithms' performance against other
algorithms\footnote{All of our experiments are run in serial on an
  Amazon C4.8xlarge machine with two Intel E5-2666 v3 chips running at
  2.90GHz; we are running Fedora 22 with kernel 4.0.4-301; we compile
  using g++ 5.1.1.}.
Our implementations represent the $i$-th letter of a permutation in
the $i$-th nibble of a 64-bit integer, allowing for permutations of
size up to 16. However, in our released code (\github), one can choose
the settings-option of allowing for larger permutations.

In Section \ref{subsectestcount}, we test our algorithm for counting
$\Pi$-hits in $S_n$ against the generate-and-check algorithm.

In Section \ref{subsectestav}, we test our algorithm for finding
$|S_n(\Pi)|$ against the naive generate-and-check algorithm and
PermLab's more sophisticated generate-and-check algorithm. Along the
way, we re-implement PermLab's algorithm, introducing optimizations
resulting from our 64-bit representation of a permutation, and
increasing efficiency for large sets of patterns. The difference in
performance between PermLab's and our algorithm is most clear for
large sets of large patterns; this is important because for large
patterns, $S_n(\Pi)$ likely often only becomes combinatorially
interesting when there are sufficiently many patterns to incur natural
structure.

In Section \ref{subsectestcompress}, we test both of our algorithms
for PPA and PPC against algorithms introduced by Inoue, Takashisa, and
Minato \cite{inoue14, inoue14followup}. Their algorithms use a
compression technique to get extremely good performance in certain
cases. We suggest directions of future work for integrating those
techniques into our algorithm for generating $S_n(\Pi)$.

Our implementations and tests are available at \github.

\subsection{Implementations counting $\Pi$-hits in each permutation in $S_n$}\label{subsectestcount}
In this section, we compare our algorithm for counting $\Pi$-hits
in each $n$-permutation to the generate-and-check algorithm.


In addition to implementing our own algorithm, we tried to implement
an efficient generate-and-check implementation for comparison. Suppose, for
simplicity, that $\Pi = \{\pi\}$ with $\pi \in S_k$. Given $\tau \in S_n$,
$\pi \in S_k$, and $j \in \{1, \ldots, k\}$, define $\beta_j(\tau, \pi)$
to be the set of $j$-letter subsequences of $\tau$'s ${n - k + j}$-upfix
that are order-isomorphic $\pi$'s $j$-upfix. Observe that $\beta_k(\tau,
\pi)$ is simply the set of $\Pi$-hits in $\tau$. Our generate-and-check
algorithm constructs $\beta_{j+1}(\tau, \pi)$ out of $\beta_{j}(\tau, \pi)$
by looping through the options for which letter to add as the smallest
letter in the new $(j+1)$-letter subsequence. The algorithm runs in
time
\begin{equation}\label{eqbrutetime}
  \Theta\left(n!  \sum_{i = 1}^k \frac{{{n - k + i} \choose
      i}}{(i-1)!}\right),
\end{equation}
since we spend $\Theta(1)$ time on an $i$-letter subsequence $s$ in
$w$ exactly when $s$ does not use any of the letters $\{1 , \ldots, k
- i\}$ and the largest $i-1$ letters of $s$ are order-isomorphic to
the $(i-1)$-upfix of $\pi$.

Our generate-and-check implementation easily extends to any $\Pi
\subseteq S_k$; in particular, we use a straightforward variant of the technique in
Appendix A to check subsequence membership in $\beta_{j+1}(\tau,
\Pi)$ in $O(1)$ time (using information about previous subsequences),
even when $|\Pi| \not\in O(1)$. The resulting algorithm runs in time
\begin{equation}\label{eqbrutetime2}
  \Theta\left(n!  \sum_{i = 1}^k k_i\frac{{{n - k + i} \choose
      i}}{(i-1)!}\right),
\end{equation}
where $k_i$ is the number of distinct elements of $S_i$ order-isomorphic to the $i$-upfix of some pattern in $\Pi$.

\fig{ \subfig{\tablebrutecntsingle}{Generate-and-check algorithm}
  \subfig{\tableuscntsingle}{Our algorithm}
  \subfig{\tablePiDDcntsingle}{\pidd-based algorithm}
}{figcntsingle}{Time in seconds to find each $\Pi$-hit in each
  permutation in $S_{\le n}$ with $n \in [8,16]$ and for $\Pi$
  containing a single pattern of length $k$ from the set $\{231, 2431,
  24531, 246531\}$.}

\fig{ \subfig{\tablebrutecntmultiple}{Generate-and-check algorithm}
  \subfig{\tableuscntmultiple}{Our algorithm}
  \subfig{\tablePiDDcntmultiple}{\pidd-based algorithm}
}{figcntmultiple}{Time in seconds to count for each $\tau\in S_{\le n}$
  the number of $k$-letter sequences containing a $231$ pattern.}

For single patterns $\pi$, our algorithm runs in $\Theta(n!)$ time,
regardless of $\pi$; in fact, its runtime is almost exactly
proportional to $n!$ (Figure \ref{figcntsingle}). In contrast, the
generate-and-check algorithm suffers as $n$ grows, as suggested in Equation
\eqref{eqbrutetime} (Figure \ref{figcntsingle}). Additionally, our
algorithm scales better to large sets $\Pi \subseteq S_k$, running in
time $O(n!k)$ time (in comparison to Equation \ref{eqbrutetime2}). As
an example, we compute the number of $k$-letter subsequences
containing $231$-pattern in each permutation $S_n$ (Figure
\ref{figcntmultiple}).

\subsection{Implementations computing $|S_n(\Pi)|$}\label{subsectestav}

In this section we compare our pattern-avoidance algorithm to the
naive generate-and-check algorithm and the more sophisticated algorithm of
PermLab.

We implemented our $O(|S_{\le n - 1}(\Pi)| \cdot k)$-time and
$O(n^k)$-space algorithm for counting $|S_1(\Pi)|, \ldots, |S_{n}(\Pi)|$, as well as a naive
generate-and-check algorithm implementation, optimizing both for
performance.

The naive generate-and-check algorithm runs as follows. Let $T_n$ be
the tree of permutations in $S_{\le n}$ such that the children of $\tau
\in S_k$ are each option for $\tau\uparrow^{i}$. Define $C(v)$ to be the
set of children of a node $v$ and $F(v)$ to be the parent. The
generate-and-check algorithm performs a depth-first search on $T_n
\cap S_{\le n}(\Pi)$, visiting a node's children only if the node
itself avoids $\Pi$. In order to determine whether a permutation
$\tau$ avoids $\Pi$, the algorithm 
\iftrue
uses the same logic as in
the previous section.
\fi

The best publicly available code for computing $S_n(\Pi)$, however, is
PermLab, which makes several clever changes to the naive
generate-and-check algorithm in order to hide its asymptotics for
small $n$ \cite{PermLab}. PermLab performs a depth-first search of
$T_n \cap S_{\le n-1}(\Pi)$, computing at a given node $v$ whether
each $c \in C(v)$ avoids $\Pi$. However, since PermLab only visits
nodes avoiding $\Pi$, it only needs to check the children of a node in
$S_j$ for $\Pi$-hits involving $j + 1$. At the same time, when
visiting a node $v \in S_j$, PermLab remembers for each $x \in
C(F(v))$ whether $x$ avoids $\Pi$. Using this information, PermLab can
quickly determine for each $c \in C(v)$ whether $c$ has a $\Pi$-hit not
involving the letter $j$. Thus PermLab needs only search through
brute force for $\Pi$-hits in $c$ involving both $j+1$ and $j$.

We re-implemented Permlab's algorithm, making optimizations specific to
our representation of permutations as 64-bit integers. We also
eliminated some wasted work by carefully examining only permutation
subsequences which could potentially form the upfix of a $\Pi$-hit;
in particular, we filter out subsequences which include a letter too
small to allow for the rest of the $\Pi$-hit to appear after the
upfix.

To search for a $\Pi$-hit in a permutation, PermLab searches
independently for each $\pi \in \Pi$ until it succeeds or concludes
the permutation avoids $\Pi$. However, this scales poorly to handling
large sets of patterns, and allows for performance to be affected by
the order patterns appear in $\Pi$. Instead, just as we did for our
generate-and-check implementation, we use a straightforward variant of
the technique from Appendix A to check whether a subsequence is
order-isomorphic to an upfix of \emph{any} $\Pi$-hit in amortized
constant time (using information about previous subsequences). This
leads to significant speedup when there are many shared upfixes among
the permutations in $\Pi$. On the other hand, if $|\Pi| = 1$, then the
overhead of using the small hash table required for the technique from
Appendix A leads to a slight slowdown. In order to demonstrate the
difference, we implement both variants, calling the
small-set-optimized version (using PermLab's scheme) V1 and the
large-set-optimized version V2.

\fig{ \subfig{\tablebruteavsingle}{Naive generate-and-check algorithm}
  \subfig{\tablehybridavsingle}{V2 algorithm}
  \subfig{\tableVavsingle}{V1 algorithm}
  \subfig{\tablepermlabavsingle}{PermLab}
  \subfig{\tableusavsingle}{Our Algorithm}
  \subfig{\tablePiDDavsingle}{\pidd-based algorithm}
}{figavsingle}{Time in seconds to compute $|S_n(\Pi)|$ with $n \in
  [8,16]$ and for $\Pi$ containing a single pattern of length $k$ from
  the set $\{231, 2431, 24531, 246531\}$.}

In Figure \ref{figavsingle} (the final subfigure of which is discussed
later), we compare the performances of the algorithms handling single
patterns\footnote{We choose not to use identity patterns, since they are
  likely to yield abnormal performance for particular algorithms.},
including V1, V2, and the original PermLab. While V1 performs slightly
faster than V2 in this experiment, it should never perform more than a
constant factor faster. Indeed, to confirm that a permutation is an
avoider, V1 must examine every permutation subsequence which V2 does;
and to discover that a permutation is not an avoider, V1 is expected
to look at at least as many sequences as V2, sometimes re-examining
sequences because patterns share a upfix. Thus the only speedup comes
from not using a small hash table to store pattern
upfixes.

Whereas the naive generate-and-check algorithm's disadvantage grows
with $n$, Permlab's algorithm appears to largely hide its asymptotic
disadvantage for single patterns. 

Shortly, we shall see that the asymptotic difference between PermLab's
algorithm and our algorithm is more pronounced when $\Pi$ is a large
set of patterns. This occurs for two reasons. First, since $S_n(\Pi)$
grows slower, our experiments can access larger $n$, at which point
asymptotics matter more. Second, for a fixed size of patterns, having
more patterns leads to a larger constant behind the cost of detecting
pattern-avoidance for PermLab's algorithm, so that its asymptotic
nature is difficult to hide even for small $n$. Indeed, suppose we
were to compute $S_n(\pi)$ for some $\pi \in S_k$. Then the average
number of subsequences checked by PermLab to detect that a permutation
in $S_n$ avoids $\pi$ might be something like $$\sum_{i = 0}^{k - 2}
\frac{{{n - k + i} \choose {i}}}{(i + 1)!},$$ assuming each
permutation in $S_i$ is equally likely to appear as a given
$i$-subsequence. But (using V2) if $\Pi$ is a large set of patterns
with $k_i$ distinct $i$-upfixes, this becomes
\begin{equation}
  \sum_{i = 0}^{k - 2}{ k_{i+2}\frac{{{n - k + i} \choose
        i}}{(i+1)!}}.
 \label{eqspec}
\end{equation}
If $k_{i+2}$ is within an order of magnitude of $(i+2)!$, then the
$1/(i+2)!$ term no longer hides the asymptotics for small $n$. In
contrast, our algorithm runs in time $O(|S_{n-1}|nk)$ regardless of
$|\Pi|$.

\fig{ \subfig{\tablebruteavmultiple}{Naive generate-and-check algorithm}
  \subfig{\tablehybridavmultiple}{V2 algorithm}
  \subfig{\tableusavmultiple}{Our Algorithm}
}{figavmultiple}{Time in seconds to compute $|S_n(X_k(231))|$.}

\fig{\tablepercentonavoiders}{figavoiderwork}{Fraction of
  permutation-subsequences viewed by V2 that lie in permutations in
  $S_n(X_k(123))$, rather than in non-avoiders.}

In Figure \ref{figavmultiple}, we show algorithm performance for large
sets of patterns. Let $X_k(231)$ be the set of permutations in $S_k$
containing a $123$-hit. Then $S_n(X_k(231))$ contains the permutations
in $S_n$ with no $k$-letter subsequences containing any
$231$-patterns; of course for $n \ge k$ this is simply $S_n(231)$,
which has size the $n$-th Catalan number $C_n$. Unlike our algorithm,
V2 and the naive generate-and-check algorithm do not scale well to
large sets of large patterns. By computing $S_n(X_k(231))$ for a fixed
$k$, we can see how each algorithm performs for varying $n$ and a
fixed large set of patterns in $S_k$. At the same time, by computing
$S_n(X_k(231))$ for a fixed $n$, we can see how each algorithm's
performance changes when we use a larger set of larger patterns to
solve the exact same pattern-avoidance problem. For this experiment,
we use our V2-variant of PermLab, since it is far more suitable for a
large set of patterns. Indeed, while the V1 variant may compute
$S_{12}(123456)$ more than twice as fast as V2, it computes
$S_{12}(X_6(231))$ more than ten times slower (11.7 seconds for V1
versus .97 seconds for V2). In fact, while V1 is ever at most some
constant times faster than V2, V2 can be arbitrarily faster than V1
for large sets of patterns.

Let us take a moment to better understand V2's performance
characteristics in Figure \ref{figavmultiple}. The algorithm's runtime
is essentially proportional to the total number of permutation
subsequences it examines to determine whether permutations are
avoiding, taking between 12 to 13 nanoseconds on average per
subsequence (based on time per subsequence when $k = 6$). 
Although there are several times more non-avoiding permutations tested
than avoiding ones, most of the time is devoted to the the avoiding
ones, and the ratio of work for avoiding versus non-avoiding checks
stays relatively constant (Figure \ref{figavoiderwork}). Thus for a
given $k$, the runtime scales according to the number of avoiders
times the average number of subsequences checked per avoider. The
latter is estimated by Equation \ref{eqspec}.

Note, however, that Equation \eqref{eqspec} assumes that the
$i$-letter subsequences of permutations in $S_n(\Pi)$ are
equi-distributed among $S_i$. A good example where this is not the
case is when computing $S_n(X_3(231))$. Here $12$ is the only valid
pattern 2-upfix, but $12$ and $21$ are \emph{not} equally likely to be
formed by $n-1$ and $n$ in a $231$-avoider. In particular, if $n-1$
precedes $n$, then $n$ must be in the final position. As a result,
Equation \eqref{eqspec} overestimates the average work done per
avoider by a bit less than two thirds. However,
Equation \eqref{eqspec} is more accurate for larger $k$, with percent
error 15.9\%, 4.9\%, and 1.7\% for $k = 4, 5, 6$ respectively and $n =
16$. More importantly, Equation \eqref{eqspec}'s accuracy is
relatively static, with the ratio of the actual work done to the
predicted work done dropping from $n = 10$ to $n = 16$ by $.06$,
$.015$, $.0005$, and $-.0002$ for $k = 3, 4, 5, 6$ respectively. As a
result, the work to generate $S_n(X_k(231))$ from $S_{n-1}(S_k(231))$
is proportional to $$C_n \sum_{i = 1}^{k - 2}{ k_{i+2}\frac{{{n - k +
i} \choose i}}{(i+1)!}}.$$

Equation \eqref{eqspec}'s accuracy was also relatively static for the
single pattern-case tested in Figure \ref{figavsingle}. The largest
drop in the ratio of work done to work anticipated by Equation
\eqref{eqspec} from $n = 8$ to $13$ is .88 dropping to .69 for $k =
5$. As a direction for future work, we suggest studying PermLab's (and
the naive generate-and-check algorithm's) performance further. For example,
is Equation \eqref{eqspec} ever an underestimate for some $\Pi
\subseteq S_k$? For a given $\Pi \in S_k$, is Equation \eqref{eqspec}'s
error bounded by some constant $c$, possibly depending on $\Pi$?


There are downsets of permutations where the difference in
algorithm performances might be much more extreme, even for single
patterns. For an example, one could consider permutations with
inversion number bounded above by some constant, and use a pattern
with few inversions. Indeed, in any case where we are interested in a
downset of permutations, many of which contain numerous
hit-upfixes, the contrast between the algorithmic performances would
be highlighted.

\subsection{In Comparison with \pidd based Algorithms}\label{subsectestcompress}

In 2013, Inoue, Takahisa, and Minato, introduced an algorithm for
generating $S_n(\Pi)$ which, although asymptotically mysterious, runs
very fast in certain cases \cite{inoue14}. Their algorithm represents
sets of permutations in a data structure called a \pidd, which in
practice compresses sets of related permutations well. They then use
set operations, in addition to other select operations easily
performed on a \pidd, in order to construct the \pidd \ for
$S_n(\Pi)$. If the \pidd's compression algorithm works sufficiently
well, the algorithm can potentially run in less than $|S_n(\Pi)|$
time. On the other hand, with poor compression, the algorithm could
perform far worse than the naive generate-and-check algorithm.

\fig{
  \begin{tabular}{l | l l l l}
    Alg $\backslash$ Set-Size & 1 & 2 & 3 & 4 \\ \hline
    Our Alg.                  & 99.765 & 3.731 & 0.197 & 0.049 \\
    \pidd-Based               & 30.648 & 28.377 & 18.328 & 32.820  \\
  \end{tabular}
}{figpiddscaling}{Time in seconds to generate $S_{15}(1234)$, $S_{15}(1234, 2341)$, $S_{15}(1234, 2341, 3412)$, $S_{15}(1234, 2341, 3412, 4123)$ respectively.}

In Figure \ref{figavsingle}, we compare the \pidd-based algorithm with
our algorithm for computing $S_k(\pi)$ for $\pi \in \{231, 2431,
24531\}$ and $n$ varying. While the \pidd-based algorithm runs
extremely fast for $|\Pi| = 1$, it performs far worse for sets of
multiple patterns. In particular, as $|\Pi|$ increases, the time to
compute $S_n(\Pi)$ tends to stay roughly constant as $|\Pi|$ grows,
instead of rapidly shrinking with $|S_n(\Pi)|$ as is the case for our
algorithm.  For an example of this, see Figure \ref{figpiddscaling}. This is
possibly because the \pidd-based algorithm works by generating the
non-avoiders and subtracting those from $S_n$, rather than directly
generating the avoiders.

Observe that Proposition \ref{propfundamental} can be rewritten in
terms of set operations.  Given a permutation $s$, define
$S\uparrow^i_j$ to be the permutation obtained by inserting $(j -
0.5)$ in position $i$ of $s$, and then standardizing the result to a
permutation. For example, $12345678\uparrow^2_5 = \st(1(4.5)2345678) =
152346789$. In turn, given a set of permutations $A$, define
$A\uparrow^i_j$ to be $\{s\uparrow^i_j | s \in A\}$. Then
Proposition \ref{propfundamental} yields the following proposition.
\begin{prop}\label{propavoidre}
Let $\Pi$ be a set of permutations, the largest of which is size
$k$. Then for $n > k$, $$S_n(\Pi) = \cap_{j = 1}^{k + 1} \cup_{i =
  1}^{n}S_{n-1}(\Pi)\uparrow^i_j.$$
\end{prop}

Thus it would be an interesting direction of future work to
efficiently implement the $\uparrow^i_j$ operation for sets
represented using \pidd. Using this, our PPA algorithm could
potentially be re-implemented using \pidd's with runtime in practice
less than $\Theta(|S_n(\Pi)|)$, even for large $\Pi$.

In 2014, Inoue, Takahisa, and Minato extended their algorithm to count
$\Pi$-hits in each $\tau \in S_n$ \cite{inoue14followup}. In
particular, they build the \pidd \ for the set of permutations with $i$
$\Pi$-hits for each $i$. In Figure \ref{figcntsingle}, we compare the
runtime-performance of the \pidd-based algorithm to our own for
single patterns $\pi \in \{231, 2431, 24531\}$, as well as to an
optimized generate-and-check implementation; this time, our algorithm
tends to have the edge. Additionally, unlike our algorithm, which runs
in time $O(n!k)$ regardless of $|\Pi|$, the \pidd-based algorithm
tends to scale approximately linearly in terms of $|\Pi|$. This can be
seen, for example, in Figure \ref{figcntmultiple}, in which our PPC
algorithm, the \pidd-based PPC algorithm, and the generate-and-check
implementation are tested on the pattern set $\{\pi \in S_k \mid
\pi \text{ contains a }231\text{-hit}\}$..

There are many interesting questions still to be asked about the
\pidd-based algorithms. Can they be extended to apply to a
downset of permutations rather than $S_n$? Can theoretical bounds be
proven for their worst-case runtime performances?

\section{Automatically generated conjectures on pattern avoidance}\label{secconjectures}

Past research enumerating $|S_n(\Pi)|$ has tended to focus on small
$|\Pi|$. Do larger sets of patterns also yield interesting number
sequences? In this section, we address this question by mass-computing
$|S_5(\Pi)|, \ldots, |S_{16}(\Pi)|$ for every choice of $\Pi \subseteq
S_4$ satisfying $|\Pi| > 4$,\footnote{On \github, we have also posted
  the same computations for $|\Pi| \in \{1, 2, 3, 4\}$.} and then
searching for the resulting sequences in OEIS \cite{oeis}. Taking out
sequences with cubic or smaller growth, we then filter these to 32,019
OEIS matches, enumerated by 446 OEIS sequences. Finally, we filter out
sets of patterns $\Pi$ for which $S_n(\Pi)$ can be enumerated via the
insertion-encoding technique \cite{insertionencoding,
  insertionencoding2}, and OEIS sequences which we can easily identify
to have been previously associated with permutation pattern
avoidance. We are left with 10 OEIS sequences which conjecturally
enumerate 82 pattern-avoidance problems, none of which can be
explained by insertion encodings. Additionally, there are 22 OEIS
sequences which each enumerate at least one pattern-avoidance problem
which cannot be explained by insertion encodings, as well as
enumerating some that can be. We have released all of our code and
data, including the number sequences $S_5(\Pi), \ldots, S_{16}(\Pi)$
for all $\Pi \subseteq S_4$, at \github.

Because making millions of requests to OEIS is not practical, we
downloaded a local copy of OEIS and built a rudimentary lookup
program. In particular, we define the \emph{OEIS match} of a number
sequence $s$ as the smallest-indexed OEIS sequence which can be
shifted no more than 14 positions to the left in order to obtain $s$.

We computed $|S_5(\Pi)|, \ldots, |S_{16}(\Pi)|$ for each $\Pi
\subseteq S_4$ with $|\Pi| > 4$, starting at $n = 5$ because
$|S_4(\Pi)|$ depends only on $|\Pi|$. There
are 2,137,358 such distinct $\Pi$ up to basic symmetry; in
particular, a given choice of $\Pi$ is equivalent to $\{f(\pi)
| \pi \in \Pi\}$ where $f$ is either the inverse, complement, or
reverse function. In fact, many more $\Pi$ still are likely subtly
equivalent, with only 64,211 distinct sequences appearing among the
2,137,358 sequences computed (i.e., there appear to be 64,211
Wilf-classes).

Next, we checked which sequences have OEIS matches. Surprisingly,
1,412,002 of the 2,137,358 $\Pi$ tested have OEIS matches, distributed
among 1386 distinct OEIS sequences. Aiming to filter out the less
interesting sequences, we went on to throw away matches for sequences
growing at constant rate (826,003 matches attributed to 290 OEIS
sequences), linear rate (391,047 matches attributed to 291 OEIS
sequences), quadratic rate (147,684 matches attributed to 264 OEIS
sequences), or cubic rate (15,249 matches attributed to 95 OEIS
sequences).\footnote{In particular, we considered $|S_5(\Pi)|, \ldots,
  S_{16}(\Pi)$ to be degree $\le k$ if its $k$-th difference was zero
  by $n = 10$.} 



We are left over with 32,019 OEIS matches, attributed to 446 distinct
OEIS sequences. It is natural to wonder how many of these matches are
false alarms. In order to test this, we examined whether running the
same computations for $n$ up to thirteen, fourteen, or fifteen,
instead of sixteen, introduces false OEIS matches which are not
revealed until $n$ gets to sixteen. Surprisingly, we have to go all
the way back to $n = 13$ before any false matches are introduced, at
which point the OEIS sequence A246878 is incorrectly paired with two
pattern-sets\footnote{At first glance, matches with sequence A133641
  also seem to appear at $n = 13$ but then disappear at $n =
  14$. However, this is simply because the OEIS sequence does not have
  enough terms in its entry. We artificially added more to our local
  copy of OEIS.}.


Out of the 446 remaining OEIS sequences, only 24 sequences
(corresponding with 251 pattern sets) are obviously already connected
to pattern avoidance based on their OEIS entries, though this likely
misses some less well-advertised results. After removing these, our
final step is to remove sets of patterns $\Pi$ for which $S_n(\Pi)$
can be completely enumerated using the insertion encoding technique
\cite{insertionencoding, insertionencoding2}. The author would like to
thank Jay Pantone and his PermPy package \cite{permpy} for pointing
these sets of patterns out.

We are left with 32 OEIS sequences conjecturally enumerating a total
289 pattern-avoidance problems, none of which are easily solved by
insertion encodings. Some of these OEIS sequences seem quite
interesting, many with combinatorial interpretations or linear
recurrences, and warrant further individual attention as stand-alone
conjectures.

Ten of the sequences are particularly interesting in that they match
\emph{only} with pattern-avoidance problems which cannot be solved
with insertion encodings, rather than with some problems which can and
some which cannot. We now list these 10 OEIS sequences, along with a
sample $\Pi$ yielding each of them. For each sequence, we provide: (1)
OEIS number and (slightly edited) OEIS brief entry; (2) Number of
pattern-sets matching to sequence; (3) Example matching pattern-set.

\begin{enumerate}

\match{A228180 The number of single edges on the boundary of ordered trees with n edges.
\\ G.f. is $(x\cdot C+2x^3\cdot C^4)/(1-x)$ where $C$ is the generating function for the Catalan numbers.}
      {11}
      {2413 4132 1432 1342 1324}
\match{A035929 Number of Dyck $n$-paths starting $U^mD^m$ (an $m$-pyramid), followed by a pyramid-free Dyck path.}
      {14}
      {2143 3142 1432 1342 1324}
\match{A071721  $\frac{6n}{(n + 1)(n + 2)} \binom{2n}{n}$.}
      {6}
      {2431 4132 1432 1342 1324 1423}
\match{A071717 Expansion of $(1+x^2C)C^2$, where $C$ is the generating function for Catalan numbers.}
      {7}
      {2431 3142 4132 1432 1342 1324 1423}
\match{A071726 Expansion of $(1+x^3C)C$, where $C$ is the generating function for Catalan numbers.}
      {6}
      {2431 2413 3142 4132 1432 1342 1324 1423}
\match{A071742 Expansion of $(1+x^4C)C$, where $C$ is the generating function for Catalan numbers.}
      {3}
      {2431 2143 3142 4132 1432 1342 1324 1423 1243}
\match{A000778 $C(n) + C(n + 1) - 1$, where $C(n)$ is the $n$-th Catalan number.}
      {24}
      {2431 3142 4132 1432 1342 1324}
\match{A109262 A Catalan transform of the Fibonacci numbers.}
      {4}
      {2413 4132 1432 1342 1423}
\match{A119370 G.f. satisfies $A(x) = 1 + xA(x)^2 + x^2(A(x)^2 - A(x))$.}
      {3}
      {2413 3142 1432 1342 1423}
\match{A124671 Row sums of a triangle generated from Eulerian numbers. \\ G.f. equals $x(1-3x+3x^2)/((1-2x)(x-1)^4)$.}
      {4}
      {2341 2134 3412 3124 1342 1324 4123 1243}

\end{enumerate}

For conciseness, we list the other 22 OEIS sequences only by their
names: A095768, A258121, A048487, A132738, A190050, A014833, A101945,
A094864, A097813, A132753, A048495, A052544, A099857, A250778,
A005592, A008776, A128543, A054492, A000918, A032908, A027994, and
A135854. More information on each of these and the corresponding sets of patterns can be found at \github.

We conclude the section with two additional small observations from
our computations. (1) The powers of two appear to enumerate $S_n(\Pi)$
for 69 different $\Pi \subseteq S_{4}$ with $|\Pi \ge 4|$, up to
trivial Wilf-equivalence. Interestingly, exactly one of these sets,
$\Pi = \{4231, 2143, 3412, 3142, 1432, 1324, 1423, 1243\}$, cannot be
handled with the insertion-encoding technique. (2) If we examine sets
$\Pi \subseteq S_4$ of size 4, then one OEIS sequence shows up which
does not appear to be already associated with pattern-avoidance. In
particular, A254316 conjecturally enumerates $S_n(3412, 3142, 4132)$
and $S_n(3142 1342 1243)$.

\section{Conclusion}\label{secconclusion}

In this paper, we provided the first provably fast algorithms for
constructing $S_n(\Pi)$ and for counting $\Pi$-hits in each $\tau \in
S_n$. Surprisingly, even though detecting \emph{whether} a permutation
contains a pattern is NP-hard \cite{bose98}, detecting \emph{which}
permutations contain that pattern is polynomial time per permutation.

Our computationally-generated data and conjectures from Section
\ref{secconjectures} seem particularly ripe for future analysis. It would
also be interesting to run additional large-scale computations to
learn about trends in pattern-avoidance. Our algorithm opens new doors
for computing $S_n(\Pi)$ in cases where $\Pi$ contains a large number
of large patterns. This makes large-scale computations finally
feasible for families of pattern-avoidance relations involving
patterns of size six, seven, and eight. In addition, our $O(n!k)$
algorithm for counting $\Pi$-hits in each $\tau \in S_n$ also brings
previously unobtainable computations within reach.

Our investigation prompts several directions for future algorithmic
work. Can Algorithm \ref{algbuildavoiders}'s runtime be improved to
$\Theta(|S_{\le n-1}(\Pi)|\cdot n)$ using the technique from Theorem
\ref{thmfactorial}? Can Algorithm \ref{algbuildavoiders} be efficiently
implemented to take advantage of \pidd's (Section
\ref{subsectestcompress})? Do the \pidd-based algorithms of Inoue,
Takashisa, and Minato \cite{inoue14, inoue14followup} have good
worst-case or expected runtimes?  What can be said about the accuracy
of Equation \eqref{eqspec}'s error for a given $\Pi$? 

Additionally, it would be interesting to extend our results to
vincular patterns, in which patterns may come with additional
adjacency constraints. In the rest of this section, we will present
our progress on this so far, as well as the challenges involved in
making further progress.

Vincular patterns came into the spotlight in 2000 when Babson and
Steingr\'imsson observed that essentially all Mahonian permutation
statistics can be written as a linear combination of the vincular
patterns appearing in a permutation \cite{babson00}. Just as for
traditional pattern-avoidance, relations to natural structures such as
Dyck paths and set partitions arise in the study of vincular-pattern avoidance \cite{claesson01}.

Vincular patterns come with position-adjacency constraints, meaning
certain pairs of adjacent positions in the pattern are required to
also be adjacent in the hit. In the context of our algorithms, it is
more convenient to discuss covincular patterns, however, which are
equivalent to vincular patterns and come with value-adjacency
constraints. A covincular patterns is a pair $(\pi, X)$ where $\pi \in
S_k$ and $X \subseteq \{0, \ldots, k\}$. An element $x \in X$ from
$1$ to $k-1$ indicates that the letters $x$ and $x + 1$ must be
represented by adjacently-valued letters in any pattern occurrence. If
$0 \in X$ (resp. $k \in X)$, then the smallest (resp. largest) letter
in any pattern-occurrence must also be the smallest (resp. largest)
letter in the entire permutation.

\begin{ex}
  The covincular pattern $(123, \{0, 2\})$ appears $n-2$ times in the
  identity permutation $e_n \in S_n$. In particular, any three letters
  $a_1, a_2, a_3$ forming the pattern must satisfy $a_1 = 1$ and $a_3
  = a_2 + 1$.
\end{ex}

Given a covincular pattern $(\pi, X)$ and a permutation $\tau$, define
$P_i(\tau)$ to be the number of $(\pi, X)$-hits in $\tau$ using the entire
$i$-upfix of $\tau$. The following proposition extends
Proposition \ref{lemPi} to the case where $\Pi$ comprises a single
covincular pattern.

\begin{prop}\label{propvincular}
  Let $\tau \in S_n$. Let $(\pi, X)$ be a covincular pattern. Then
  \[ P_i(\tau) = \left\{
    \begin{array}{ll}
      P_{i+1}                                & \mbox{if $i < n$, $i \le |\pi|$, and $i \in X$,} \\
      P_{i+1}(\tau) + P_i (\tau \downarrow_{i + 1}) & \mbox{if $i < n$, $i \le |\pi|$, $i \not\in X$,} \\
      1                                      & \mbox{if $i=n$ and $\tau\in\Pi$, and} \\
      0                                      & \mbox{otherwise.}
  \end{array}
 \right\} \]
\end{prop}
\begin{proof}
  Cases (2)--(4) follow just as in the proof of
  Proposition \ref{lemPi}. Suppose $i < n$, $i \le |\pi|$, and $i \in
  X$. If $i = |\pi|$, then since $i < n$ and $i \in X$ we see that
  $P_i(\tau) = 0$, which Case (4) tells us is also the value of
  $P_{i+1}(\tau)$, as desired. On the other hand, if $i < |\pi|$, then
  since $i \in X$, any copy of $\pi$ in $\tau^{-1}$ using the
  $i$-upfix of $\tau$ must also use the $i+1$-upfix of $\tau$. Thus
  $P_i(\tau) = P_{i+1}(\tau)$.
\end{proof}

Using this recurrence, analogues of Theorems \ref{thmdownset} and \ref{thmfactorial} follow with only slightly modified proofs.
\begin{thm}\label{thmvincular1}
  Let $(\pi, X)$ be a covincular pattern. Given a downset $D \subseteq
  S_{\le n}$ and $d^{-1}$ for each $d \in D$, one can count $(\pi,
  X)$-hits in $\tau$ for each $\tau \in D$ in $O(|D|\cdot |\pi|)$
  time.
\end{thm}
\begin{proof}
  If one modifies Algorithm \ref{algcounthits} to use the recurrence
  from Proposition \ref{propvincular} on $\tau$ rather than the
  recurrence from Proposition \ref{lemPi} on $\tau$, then the proof
  follows just as for Theorem \ref{thmdownset}.

\end{proof}

\begin{thm}\label{thmvincular2}
  Let $(\pi, X)$ be a covincular pattern. Then the number of $(\pi, X)$-hits
  in each $\tau \in S_n$ can be computed in $\Theta(n!)$ time,
  regardless of $|\pi|$.
\end{thm}
\begin{proof}
  The result follows using the same technique as in the proof of
  Theorem \ref{thmfactorial}. In particular, when applying the
  recursion from Proposition \ref{propvincular} to compute $P_i(\tau)$
  for some $\tau \in S_n$, one checks whether the $i$-upfix of $\tau$
  is order-isomorphic to the $i$-upfix of $\pi$. If the two are not
  order-isomorphic, $P_i(\tau)$ must be zero.
\end{proof}

Theorem \ref{thmvincular2} shows that we can count $(\pi, X)$-hits for each
covincular permutation in $S_{\le n}$ in $\Theta(n!)$ time. By
considering each pattern in $\Pi$ separately, this extends to an
algorithm for counting $\Pi$-hits for any set $\Pi$ of covincular
permutations in $O(n!|\Pi|)$ time.

It is still an open problem, however, to quickly build $S_{\le n}(\Pi)$ if
$\Pi$ comprises covincular patterns. The difficulty in this comes from
the fact that $S_{\le n}(\Pi)$ needs not be a downset in this
case. Indeed, removing a letter from an avoider $\tau$ may introduce a
covincular pattern which was not previously present. For example, the
permutation $1342$ does not contain a covincular $(123, \{1\})$ pattern, but removing
$2$ results in a permutation which does.

One special case of a covincular pattern is when $X = \{1, \dots, k -
1\}$, meaning that every pair of adjacently valued letters in the
pattern must also be adjacently valued in any occurrence of the
pattern. This is what's known as a \emph{consecutive pattern}. For
consecutive patterns $\pi$, Theorem \ref{thmvincular1} counts
$\pi$-hits in a downset $D$ in $O(|D| \cdot |\pi|)$ time (assuming $d^{-1}$ is
known for each $d \in D$). Interestingly, in this case, the PPM
problem (detecting a $\pi$-pattern in a single permutation $\tau
\in S_n$) already has a linear time solution due to Kubica,
Kulczy{\'n}ski, Radoszewski, Rytter, and Wale{\'n}  \cite{kubica13}. A
similar result was found independently by Kim et. al.  \cite{kim14}.

\section{Acknowledgments}

This research was conducted at the University of Minnesota Duluth REU
and was supported by NSF grant 1358659 and NSA grant
H98230-13-1-0273. Machine time was provided by an AWS in Education
Research grant. The author thanks Joe Gallian for suggesting the
problem; Samuel D. Judge and David Moulton for offering advice on
directions of research; Yuma Inoue for discussing and providing
implementations of \pidd-based algorithms; and Michael Albert for his
PermLab implementation.

\section{Appendix A}

Given a set of pattern $\Pi$, a permutation $\tau \in S_n$, and the
inverse $\tau^{-1}$, a common computation is to compute for which $i$
the $i$-upfix of $\tau$ is order-isomorphic to the $i$-upfix of any
permutation $\pi \in \Pi$. It turns out that one can run this check
for all $i$ in the range $1, \ldots, r$ in time $O(r)$.

To do this, for each $i \in [r]$, we compute the standardization of
the $i$-upfix of $\tau$, and then check its membership in a hash
table\footnote{Note that there is a small preprocessing cost which
must be paid at the beginning of the algorithm to build these hash
tables.} containing the standardized $i$-upfixes of each
$\pi \in \Pi$. In fact, it turns out we can compute the
standardizations of each of the successive $i$-upfixes in constant
time. This takes advantage of the $\popcount$ instruction, which on
most modern machines obtains the number of 1s in an integer's binary
representation through a single instruction. In particular, we
maintain a bitmap $b$ (in the form of an integer) where $b[j] = 1$ if
some $k \in [n-i+1, \ldots, n]$ is in position $j$. We can then use
$\popcount$ to query how many letters in $\tau$'s $i$-upfix appear to
the right of $n-i+1$; this tells us in what position to insert $1$
into the standardized $(i-1)$-upfix in order to obtain the
standardized $i$-upfix. The insertion can then be performed using bit
hacks in constant time.

\newpage
\bibliographystyle{plain} \pagestyle{empty}\singlespace
\bibliography{paper}

\end{document}